%% file: paper.tex
\documentclass{article}

\usepackage{microtype}
\usepackage{graphicx}
\usepackage{subcaption}
\usepackage{booktabs}
\usepackage{multirow}
\usepackage{hyperref}
\usepackage{xurl}
\usepackage{tikz}
\usetikzlibrary{arrows.meta,bending,positioning,calc,decorations.pathreplacing,fit,backgrounds}

\usepackage[accepted]{icml2026}

\makeatletter
\renewcommand{\Notice@String}{\textit{Published at the Trustworthy AI4GOOD
Workshop at the 43rd International Conference on Machine
Learning (ICML 2026), Seoul, South Korea. Copyright 2026 by the
author(s).}}
\makeatother

\usepackage{amsmath}
\usepackage{amssymb}
\usepackage{mathtools}
\usepackage{amsthm}

\usepackage[capitalize,noabbrev]{cleveref}

\usepackage{enumitem}

\crefname{theorem}{theorem}{theorems}
\Crefname{theorem}{Theorem}{Theorems}
\crefname{proposition}{proposition}{propositions}
\Crefname{proposition}{Proposition}{Propositions}
\crefname{lemma}{lemma}{lemmas}
\Crefname{lemma}{Lemma}{Lemmas}
\crefname{corollary}{corollary}{corollaries}
\Crefname{corollary}{Corollary}{Corollaries}
\crefname{definition}{definition}{definitions}
\Crefname{definition}{Definition}{Definitions}
\crefname{assumption}{assumption}{assumptions}
\Crefname{assumption}{Assumption}{Assumptions}
\crefname{example}{example}{examples}
\Crefname{example}{Example}{Examples}
\crefname{remark}{remark}{remarks}
\Crefname{remark}{Remark}{Remarks}
\crefname{algorithm}{algorithm}{algorithms}
\Crefname{algorithm}{Algorithm}{Algorithms}

\usepackage{aliascnt}
\theoremstyle{plain}
\newtheorem{theorem}{Theorem}[section]

\newaliascnt{proposition}{theorem}
\newtheorem{proposition}[proposition]{Proposition}
\aliascntresetthe{proposition}

\newaliascnt{lemma}{theorem}
\newtheorem{lemma}[lemma]{Lemma}
\aliascntresetthe{lemma}

\newaliascnt{corollary}{theorem}
\newtheorem{corollary}[corollary]{Corollary}
\aliascntresetthe{corollary}

\theoremstyle{definition}
\newaliascnt{definition}{theorem}
\newtheorem{definition}[definition]{Definition}
\aliascntresetthe{definition}

\newaliascnt{assumption}{theorem}

\aliascntresetthe{assumption}

\newaliascnt{example}{theorem}

\aliascntresetthe{example}

\theoremstyle{remark}
\newaliascnt{remark}{theorem}

\aliascntresetthe{remark}

\usepackage{algorithm}
\usepackage{algorithmic}

\icmltitlerunning{ChainMark: Model-Free LLM Watermarking}

\begin{document}

\twocolumn[
  \icmltitle{ChainMark: Model-Free LLM Watermarking with\\
    Closed-Form Calibration}

  \icmlsetsymbol{equal}{*}

  \begin{icmlauthorlist}
    \icmlauthor{Chengheng Li-Chen}{epfl}
    \icmlauthor{Kyuhee Kim}{epfl}
  \end{icmlauthorlist}

  \icmlaffiliation{epfl}{École Polytechnique Fédérale de Lausanne (EPFL), Lausanne, Switzerland}

  \icmlcorrespondingauthor{Chengheng Li-Chen}{chengheng.lichen@epfl.ch}

  \icmlkeywords{AI governance, LLM watermarking, content provenance,
    EU AI Act, technical AI governance}

  \vskip 0.3in
]

\printAffiliationsAndNotice{}

\input{01_abstract}

\input{02_introduction}
\input{03_related_work}
\input{04_preliminaries}
\input{05_method}
\input{06_theory}
\input{07_experiments}
\input{08_discussion}
\input{09_conclusion}

\input{10_impact}

\bibliography{paper}
\bibliographystyle{icml2026}

\appendix
\input{11_proofs}
\input{12_reference_tables}
\input{13_reproducibility}
\input{14_run_data}

\input{15_reference_impl}
\input{16_generation_examples}

\end{document}

%% file: 01_abstract.tex
\begin{abstract}
Regulatory regimes such as the EU AI Act mandate machine-readable marking of synthetic text, but existing watermarks are ill-suited to regulator-facing audit: entropy-gated detectors require the generating model at detection time, and no scheme maps a regulator's targets (a false-positive rate, a text-length floor, and a watermark budget) to a deployer configuration in closed form. \\
We introduce \emph{ChainMark}, an active watermark whose detection needs no model: a keyed hash partitions the vocabulary into $S$ states arranged in a cycle, and at a fraction $\rho$ of positions the model is forced to emit the next state, so the text walks the cycle. A verifier replays the partition from the same key in $O(n)$ hash operations in the text length $n$. We derive a closed-form map $S^\star(n_{\min},\rho,\alpha)$ from the regulator's targets to the minimum state count, which a one-corpus recalibration pins to a 1\% false-positive rate on real text; under uniform substitution the watermark survives edits up to a fraction $\delta^\star = 1 - 1/\sqrt{2} \approx 29.3\%$ independent of $(S,\rho,n)$. Across three instruction-tuned LLMs and four domains, ChainMark retains a 72.8\% true-positive rate after Chinese back-translation, where KGW and SWEET fall below 20\% at matched budget. ChainMark thus maps a regulatory specification to a deployable configuration that a third party can audit without the model.
\end{abstract}

%% file: 02_introduction.tex
\section{Introduction}
\label{sec:intro}
EU AI Act Article~50~\citep{euaiact2024} requires machine-readable marking of AI-generated text by August 2026, but the technical substrate is unsettled. Existing watermarks have three structural problems for regulator-facing audit. \emph{(i)~Model-bound detection.} SWEET~\citep{lee2023sweet} and EWD~\citep{lu2024ewd} select or weight tokens by entropy and so need the generating language model (LM) at detection time; and the inverse-transform distortion-free watermark~\citep{kuditipudi2023robust} needs the full secret key sequence rather than a short key. \emph{(ii)~No closed-form calibration.} None of these schemes converts a regulator's targets (a target false-positive rate $\alpha$, a minimum text length $n_{\min}$, and the fraction $\rho$ of tokens to mark) into the concrete settings a deployer must apply; KGW's~\citep{kirchenbauer2023watermark} watermark strength $\delta_{\mathrm{KGW}}$ is set by hand. \emph{(iii)~Brittle robustness.} At matched budget, the true-positive rate (TPR) of KGW and SWEET falls below 20\% after Chinese back-translation, i.e.\ translating to Chinese and back (\autoref{tab:headline-aggregate}).

\input{figure-1.tex}

We propose \emph{ChainMark}, built on one idea: a secret key hashes every vocabulary token (via SHA-256) into one of $S$ states, the states are arranged in a fixed cycle, and at watermarked positions we constrain the model to emit a token from the next state, so the generated text walks the cycle by construction (\autoref{fig:teaser}). A  verifier holding the key re-derives each token's state and counts how often the walk is followed, in $O(n)$ hash operations, needing the key and tokenizer but no model.

Two properties follow. The detector's false-positive behavior is known in closed form, giving a map $S^\star(n_{\min},\rho,\alpha)$ from a regulator's targets to the minimum number of states; because real text is not uniform over the states, a one-corpus recalibration restores the target false-positive rate. Under uniform random substitution the watermark also has a robustness threshold $\delta^\star = 1 - 1/\sqrt{2} \approx 29.3\%$ independent of $S$, $\rho$, and $n$, keeping a 72.8\% TPR after Chinese back-translation. A deployer therefore configures the watermark directly from a regulatory target, and any auditor can check the output from the key without access to the model.

\paragraph{Contributions.}
\begin{enumerate}[nosep, leftmargin=*]
    \item We introduce ChainMark, a hard-constraint watermark with $O(n)$, model-free detection from the key and tokenizer (\autoref{sec:method}).
    \item We derive an operator-facing calibration map from a target false-positive rate, text-length floor, and budget to the minimum admissible state count (\autoref{thm:detection}, \autoref{app:proofs}), with an empirical recalibration that restores this target FPR on natural-language text (\autoref{thm:detection-recal}), and a robustness threshold under uniform substitution that is invariant to $(S,\rho,n)$ (\autoref{thm:robustness}).
    \item Both the detection law and the robustness threshold extend to any $k$-regular doubly-stochastic transition topology (\autoref{thm:k-regular}).
    \item We show that at matched budget, across three instruction-tuned LLMs and four domains, ChainMark retains substantially more detection signal than KGW and SWEET under translation and random-substitution attacks, at roughly $2\times$ their self-perplexity (\autoref{sec:experiments}).
\end{enumerate}

%% file: figure-1.tex
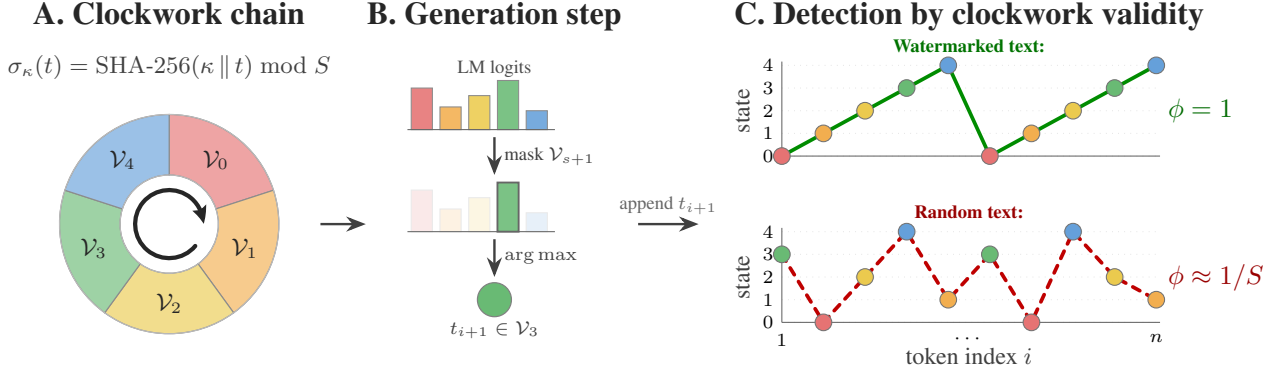
\begin{figure*}[t]
  \centering
  \begin{tikzpicture}[
    x=1cm, y=1cm,
    paneltitle/.style={font=\small\bfseries, anchor=base, inner sep=0pt,
                       text=black!85},
    paneltitlewest/.style={font=\small\bfseries, anchor=base west,
                           inner sep=0pt, text=black!85},
    axisline/.style={black!60, line width=0.4pt},
    gridline/.style={black!12, line width=0.3pt, dotted},
    pt/.style={circle, draw=black!55, line width=0.4pt,
               minimum size=2.2mm, inner sep=0pt},
    bigpt/.style={circle, draw=black!50, line width=0.45pt,
                  minimum size=4.5mm, inner sep=0pt},
    wmline/.style={line width=1.4pt, green!55!black,
                   line cap=round, line join=round},
    rndline/.style={line width=1.4pt, red!75!black, dashed,
                    line cap=round, line join=round},
    yticklbl/.style={font=\scriptsize, anchor=east, inner sep=1.5pt},
    xticklbl/.style={font=\scriptsize, anchor=north, inner sep=1.5pt},
    axislabel/.style={font=\footnotesize, color=black!75},
    score/.style={font=\normalsize\bfseries, anchor=west, inner sep=0pt},
    flowarr/.style={->, line width=0.9pt, black!75,
                    >={Stealth[length=2.6mm,width=2mm]}},
  ]
    \definecolor{c0}{HTML}{E15A5A}
    \definecolor{c1}{HTML}{F0A030}
    \definecolor{c2}{HTML}{E8C12F}
    \definecolor{c3}{HTML}{4FAE5C}
    \definecolor{c4}{HTML}{4A8FD3}

    \node[paneltitle] at (1.80, 4.95) {\large A.~Clockwork chain};
    \node[font=\footnotesize, color=black!75, anchor=base]
      at (1.80, 4.30)
      {$\sigma_\kappa(t) = \mathrm{SHA\text{-}256}(\kappa\,\Vert\,t)\bmod S$};

    \begin{scope}[shift={(1.80, 2.30)}]
      \foreach \i [evaluate=\i as \stt using {90 - \i*72},
                   evaluate=\i as \fin using {90 - (\i+1)*72}] in {0,1,2,3,4} {
        \fill[c\i!55, draw=black!45, line width=0.45pt]
          (\stt:1.45) arc[start angle=\stt, end angle=\fin, radius=1.45]
          -- (\fin:0.65) arc[start angle=\fin, end angle=\stt, radius=0.65]
          -- cycle;
      }
      \draw[black!25, line width=0.3pt] (0,0) circle (1.45);
      \draw[black!25, line width=0.3pt] (0,0) circle (0.65);
      \foreach \i [evaluate=\i as \mid using {90 - (\i+0.5)*72}] in {0,1,2,3,4} {
        \node[font=\small\bfseries, color=black!85]
          at (\mid:1.05) {$\mathcal V_{\i}$};
      }
      \draw[->, line width=1.6pt, black!85, line cap=round,
            >={Stealth[length=2.2mm,width=1.9mm,inset=0.3mm,bend]}]
        (320:0.45) arc[start angle=320, end angle=0, radius=0.45];
    \end{scope}

    \node[paneltitle] at (6.10, 4.95) {\large B.~Generation step};
    \node[font=\scriptsize, color=black!70, anchor=base]
      at (6.10, 4.30) {LM logits};
    \foreach \i/\xx/\h in {0/5.00/0.55, 1/5.38/0.30, 2/5.76/0.45,
                            3/6.14/0.65, 4/6.52/0.25} {
      \draw[fill=c\i!75, draw=black!45, line width=0.4pt]
        (\xx, 3.55) rectangle (\xx + 0.28, 3.55 + \h);
    }
    \draw[axisline] (4.95, 3.55) -- (6.90, 3.55);

    \draw[flowarr] (6.10, 3.45) -- (6.10, 2.95)
       node[midway, right, font=\scriptsize, color=black!75]
       {mask $\mathcal V_{s+1}$};

    \foreach \i/\xx/\h in {0/5.00/0.55, 1/5.38/0.30, 2/5.76/0.45,
                            3/6.14/0.65, 4/6.52/0.25} {
      \draw[fill=c\i!75, draw=black!45, line width=0.4pt,
            opacity=0.18] (\xx, 2.20) rectangle (\xx + 0.28, 2.20 + \h);
    }
    \draw[fill=c3!75, draw=black!60, line width=0.7pt]
        (6.14, 2.20) rectangle (6.14 + 0.28, 2.20 + 0.65);
    \draw[axisline] (4.95, 2.20) -- (6.90, 2.20);

    \draw[flowarr] (6.10, 2.10) -- (6.10, 1.60)
       node[midway, right, font=\scriptsize, color=black!75]
       {$\arg\max$};

    \node[bigpt, fill=c3!85] (toknext) at (6.10, 1.30) {};
    \node[font=\scriptsize, color=black!75, anchor=base]
      at (6.10, 0.85) {$t_{i+1} \in \mathcal V_{3}$};

    \def\xstep{0.55}
    \def\ystep{0.30}
    \def\plotL{9.90}
    \def\plotR{14.85}
    \def\plotw{4.95}
    \def\ploth{1.20}
    \def\pcen{0.60}
    \node[paneltitle] at (\plotL + \plotw/2, 4.95)
      {\large C.~Detection by clockwork validity};

    \def\wmbot{3.20}
    \def\wmtop{4.40}
    \node[font=\scriptsize\bfseries, color=green!45!black, anchor=base]
      at (\plotL + \plotw/2, \wmtop + 0.18)
      {Watermarked text:};

    \foreach \s in {0,1,2,3,4} {
      \draw[gridline] (\plotL, \wmbot + \s*\ystep) -- (\plotR, \wmbot + \s*\ystep);
    }
    \draw[axisline] (\plotL, \wmbot) -- (\plotR + 0.05, \wmbot);
    \draw[axisline] (\plotL, \wmbot) -- (\plotL, \wmtop + 0.05);
    \foreach \s in {0,1,2,3,4} {
      \node[yticklbl] at (\plotL - 0.05, \wmbot + \s*\ystep) {$\s$};
    }
    \node[axislabel, rotate=90, anchor=south]
      at (\plotL - 0.30, \wmbot + \pcen) {state};

    \draw[wmline]
      (\plotL + 0*\xstep, \wmbot + 0*\ystep)
      -- (\plotL + 1*\xstep, \wmbot + 1*\ystep)
      -- (\plotL + 2*\xstep, \wmbot + 2*\ystep)
      -- (\plotL + 3*\xstep, \wmbot + 3*\ystep)
      -- (\plotL + 4*\xstep, \wmbot + 4*\ystep)
      -- (\plotL + 5*\xstep, \wmbot + 0*\ystep)
      -- (\plotL + 6*\xstep, \wmbot + 1*\ystep)
      -- (\plotL + 7*\xstep, \wmbot + 2*\ystep)
      -- (\plotL + 8*\xstep, \wmbot + 3*\ystep)
      -- (\plotL + 9*\xstep, \wmbot + 4*\ystep);
    \foreach \i/\s in {0/0,1/1,2/2,3/3,4/4,5/0,6/1,7/2,8/3,9/4} {
      \node[pt, fill=c\s!85] at (\plotL + \i*\xstep, \wmbot + \s*\ystep) {};
    }
    \node[score, color=green!42!black]
      at (\plotR + 0.15, \wmbot + \pcen) {$\phi = 1$};

    \def\rndbot{1.00}
    \def\rndtop{2.20}
    \node[font=\scriptsize\bfseries, color=red!60!black, anchor=base]
      at (\plotL + \plotw/2, \rndtop + 0.18)
      {Random text:};

    \foreach \s in {0,1,2,3,4} {
      \draw[gridline] (\plotL, \rndbot + \s*\ystep) -- (\plotR, \rndbot + \s*\ystep);
    }
    \draw[axisline] (\plotL, \rndbot) -- (\plotR + 0.05, \rndbot);
    \draw[axisline] (\plotL, \rndbot) -- (\plotL, \rndtop + 0.05);
    \foreach \s in {0,1,2,3,4} {
      \node[yticklbl] at (\plotL - 0.05, \rndbot + \s*\ystep) {$\s$};
    }
    \node[axislabel, rotate=90, anchor=south]
      at (\plotL - 0.30, \rndbot + \pcen) {state};

    \node[xticklbl] at (\plotL,                 \rndbot - 0.10) {$1$};
    \node[xticklbl] at (\plotL + 4.5*\xstep,    \rndbot - 0.10) {$\cdots$};
    \node[xticklbl] at (\plotL + 9*\xstep,      \rndbot - 0.10) {$n$};
    \node[axislabel, anchor=base]
      at (\plotL + \plotw/2, \rndbot - 0.55) {token index $i$};

    \draw[rndline]
      (\plotL + 0*\xstep, \rndbot + 3*\ystep)
      -- (\plotL + 1*\xstep, \rndbot + 0*\ystep)
      -- (\plotL + 2*\xstep, \rndbot + 2*\ystep)
      -- (\plotL + 3*\xstep, \rndbot + 4*\ystep)
      -- (\plotL + 4*\xstep, \rndbot + 1*\ystep)
      -- (\plotL + 5*\xstep, \rndbot + 3*\ystep)
      -- (\plotL + 6*\xstep, \rndbot + 0*\ystep)
      -- (\plotL + 7*\xstep, \rndbot + 4*\ystep)
      -- (\plotL + 8*\xstep, \rndbot + 2*\ystep)
      -- (\plotL + 9*\xstep, \rndbot + 1*\ystep);
    \foreach \i/\s in {0/3,1/0,2/2,3/4,4/1,5/3,6/0,7/4,8/2,9/1} {
      \node[pt, fill=c\s!85] at (\plotL + \i*\xstep, \rndbot + \s*\ystep) {};
    }
    \node[score, color=red!60!black]
      at (\plotR + 0.15, \rndbot + \pcen) {$\phi \approx 1/S$};

    \draw[flowarr] (3.80, 2.30) -- (4.40, 2.30);
    \draw[flowarr] (8.00, 2.30) -- (8.80, 2.30)
       node[midway, above, font=\scriptsize, color=black!60]
       {append $t_{i+1}$};
  \end{tikzpicture}
  \caption{\textbf{ChainMark at a glance.}
    \textbf{(A)}~A keyed SHA-256 hash $\sigma_\kappa$ partitions the
    vocabulary into $S$ disjoint sets $\mathcal V_0, \dots,
    \mathcal V_{S-1}$ ordered into a clockwork cycle
    $s \to s{+}1 \bmod S$ (here $S{=}5$).
    \textbf{(B)}~At a watermarked position with current state $s$,
    the language-model logits are masked to keep only the next
    partition $\mathcal V_{s+1}$, and the next token is the masked
    $\arg\max$.
    \textbf{(C)}~Detection re-derives every token's state from the
    key $\kappa$ in $O(n)$ hash operations: a watermarked sequence walks
    the chain by construction (top, validity rate $\phi=1$), whereas
    unwatermarked text visits states uniformly and is valid only at
    rate $1/S$ in expectation (bottom).}
  \label{fig:teaser}
\end{figure*}

%% file: 03_related_work.tex
\section{Related Work}
\label{sec:related}

\paragraph{Passive detection.} DetectGPT~\citep{mitchell2023detectgpt}
exploits log-probability curvature; GPTZero~\citep{tian2023gptzero}
combines perplexity and burstiness. Both read statistical traces in
existing text, with no controllable false-positive frontier and sharp
degradation under
paraphrase~\citep{krishna2023paraphrasing,sadasivan2023can}. They
cannot serve as the technical substrate for Article~50 because neither
regulator nor deployer can set the detection regime in advance.

\paragraph{Active watermarking.} Green--red-list (KGW)
watermarking~\citep{kirchenbauer2023watermark} biases next-token
logits toward a hashed green list and detects via a $z$-test on
green-token frequency. The scheme's signal depends on next-token
entropy and degrades under paraphrase, and offers no analytic inversion
from a regulatory parameter back to a bias magnitude. ChainMark's
fingerprint $\phi$ is structurally a Bernoulli($1/S$) $z$-test,
so the calibration formula~\eqref{eq:calibration} is the analogue
KGW lacks.

\paragraph{Distortion-free schemes.}
\citet{aaronson2023semstamp}'s Gumbel-max construction and
\citet{kuditipudi2023robust}'s inverse-transform sampling watermark
are \emph{distortion-free}: they preserve the LM marginal in
expectation over the key. \citet{christ2023undetectable} prove
cryptographic indistinguishability under sufficient entropy. These
schemes beat ChainMark on quality strictly (zero expected KL versus ChainMark's
$\geq \rho \log S$). The trade is calibration: none maps a
regulatory target $(\alpha, n_{\min}, \rho)$ to a configuration in
closed form, and the inverse-transform scheme's detector needs the
full key sequence rather than a short key. ChainMark is the distorting
alternative that supplies the closed-form $S^\star$ map.


\paragraph{Entropy-adaptive schemes: gates vs.\ detectors.}
SWEET~\citep{lee2023sweet} and EWD~\citep{lu2024ewd} both restrict
watermarking signal to \emph{high-entropy} positions, but at different
layers: SWEET applies a binary gate at generation time, while EWD
weights positions by entropy at detection time, approximating the
log-likelihood ratio of \autoref{thm:hdd}. Either way, detection must
recompute per-token entropy and so needs the generating model, unlike
ChainMark. ChainMark is gate-agnostic in its specification
(high-entropy, low-entropy, and surprisal-gap gates are all
admissible); for the head-to-head in \autoref{sec:experiments} we
adopt the high-entropy gate to match SWEET at the same budget, so the
operational dial we vary is $\rho$, not the gate identity.

\paragraph{Technical AI governance.} Article~50 of the EU AI
Act~\citep{euaiact2024} and the EU's draft Code of Practice on
Transparency~\citep{eucodepractice2025} require machine-readable
marking of AI-generated content; the OECD Hiroshima
Process~\citep{oecdhiroshima2024} provides a parallel international
framework. These instruments leave the detection regime
(FPR, quality floor, robustness) implicit. ChainMark's contribution, beyond
the watermarking scheme itself, is that it surfaces these parameters
analytically, producing a policy-to-configuration map auditable by
third parties (\autoref{sec:discussion}).

%% file: 04_preliminaries.tex
\section{Preliminaries}
\label{sec:prelim}

\textbf{Notation.} Let $\mathcal{V}$ denote the vocabulary of an
autoregressive language model, $V = |\mathcal{V}|$, and write
$\mathbf{t} = (t_1, \ldots, t_n) \in \mathcal{V}^n$ for a token
sequence of length $n$. The deployer and the verifier share a secret
key $\kappa \in \{0,1\}^*$ and the same tokenizer; the verifier
never sees the language model. Let
$\mathcal{H} : \{0,1\}^* \to \{0,1\}^{256}$ be SHA-256, modelled as
a random oracle in the security analysis (\autoref{thm:security}).
We write $[m] = \{0, \ldots, m-1\}$ and $z_\alpha$ for the one-sided
standard-normal quantile at level $\alpha$.

\textbf{Threat model.} The adversary observes watermarked text and
may paraphrase, translate, edit, splice, or delete tokens before
publication; we model these as a per-token modification rate
$\delta \in [0,1)$. The adversary is computationally bounded and
has no access to $\kappa$, to the language model, or to the
deployed detector as a queryable oracle at detection time
(\autoref{sec:discussion} discusses the adaptive-oracle adversary
that lies outside this model). The verifier sees only token text and runs in $O(n)$ hash
operations.

\textbf{Detection problem and quality.} Detection is a binary
hypothesis test on text alone, $H_0$ (non-watermarked) vs.\ $H_1$
(watermarked under key $\kappa$), at a fixed false-positive rate
$\alpha$, since false positives are operationally the costly error.
The watermark distribution must remain fluent under standard greedy
or low-temperature decoding, which we control via the per-token KL
divergence to the original next-token distribution and report as
perplexity inflation (\autoref{sec:method}).

%% file: 05_method.tex
\section{The ChainMark Scheme}
\label{sec:method}

\subsection{State Partition}

For each token id $v \in \mathcal{V}$ define
\begin{equation}
  \sigma_\kappa(v) \;=\; \mathcal{H}(\kappa \,\Vert\, v) \bmod S,
\end{equation}
which deterministically maps $\mathcal{V}$ onto $S$ equivalence
classes $\mathcal{V}_s = \{v : \sigma_\kappa(v) = s\}$, $s \in [S]$. Under
the random-oracle model $\sigma_\kappa$ is uniform and independent across
distinct tokens, so $\mathbb{E}[|\mathcal{V}_s|] = V/S$; knowledge
of $\sigma_\kappa$ on queried tokens reveals nothing about unqueried ones
(\autoref{thm:security}).

\subsection{Transition Topology}
\label{subsec:topology}

A successor function $\Sigma_k : [S] \to \binom{[S]}{k}$ assigns
each state a $k$-element set of legal next states. We require
$\Sigma_k$ to be $k$-regular in both out- and in-degree;
column-regularity is essential for the null-variance and robustness
arguments (\autoref{thm:k-regular}, \autoref{lem:variance}). The
default is the \emph{clockwork} chain at $k = 1$,
$\Sigma_1(s) = \{(s{+}1) \bmod S\}$: deterministic, periodic of
period $S$, uniform stationary distribution, and the smallest random
baseline $1/S$ among $1$-regular topologies. We additionally
evaluate the \emph{soft-cycle} variant $k = 2$,
$\Sigma_2(s) = \{(s{+}1) \bmod S,\,(s{+}2) \bmod S\}$, which doubles
the per-state successor budget at the cost of raising the random
baseline to $2/S$ (\autoref{sec:experiments}).

\subsection{Generation}

At step $i$ with current state $s_i = \sigma_\kappa(t_i)$, ChainMark forms
the legal token set $\mathcal{V}_{s_i}^\star =
\bigcup_{s' \in \Sigma_k(s_i)} \mathcal{V}_{s'}$ and masks the LM
logits outside $\mathcal{V}_{s_i}^\star$ to $-\infty$. At gated
positions ChainMark picks the argmax over the masked logits; at ungated
positions it samples from the original distribution. Each step costs one LM forward pass
plus an $O(V)$ logit mask (precomputed once per state).

\begin{algorithm}[t]
  \caption{ChainMark Watermark Embedding}
  \label{alg:embed}
  \begin{algorithmic}
    \STATE {\bfseries Input:} prompt $p$, key $\kappa$, states $S$,
      successor $\Sigma_k$, gate $g(\cdot)$, max tokens $N$
    \STATE $\mathbf{t} \leftarrow \mathrm{Tokenize}(p)$;
      $s \leftarrow \sigma_\kappa(\mathbf{t}_{-1})$
    \FOR{$i = 1$ {\bfseries to} $N$}
      \STATE $p_i \leftarrow \mathrm{softmax}(\mathcal{M}(\mathbf{t}))$
      \STATE $\mathcal{V}^\star \leftarrow
        \bigcup_{s' \in \Sigma_k(s)} \mathcal{V}_{s'}$
      \IF{$g(p_i) = 1$ {\bfseries and} $p_i(\mathcal{V}^\star) > 0$}
        \STATE $t_i \leftarrow \arg\max_{t \in \mathcal{V}^\star} p_i(t)$
      \ELSE
        \STATE $t_i \sim p_i$ \hfill{\footnotesize// free sampling at ungated positions}
      \ENDIF
      \STATE $\mathbf{t} \leftarrow \mathbf{t} \,\Vert\, t_i$;
        $s \leftarrow \sigma_\kappa(t_i)$
      \IF{$t_i = \texttt{EOS}$} \STATE {\bfseries break} \ENDIF
    \ENDFOR
    \STATE \textbf{return} $\mathrm{Decode}(\mathbf{t})$
  \end{algorithmic}
\end{algorithm}

\subsection{Watermark Budget and Entropy Gate}
\label{subsec:gates}

The mask is applied only at a subset of positions; the
\emph{watermark budget} is $\rho = \mathbb{E}_i[g_i] \in [0,1]$. Let
$H_i = -\sum_t p_i(t) \log p_i(t)$ be the next-token entropy and
$\Delta_i = p_i^{(1)} - p_i^{(2)}$ the gap between the top two
probabilities.

\begin{definition}[Gate family]
  \label{def:gate}
  $G_{\mathrm{all}}{:}\; g_i = 1$;\;
  $G_{H_{\mathrm{high}}}(\tau){:}\; g_i = \mathbb{1}[H_i > \tau]$;\;
  $G_{H_{\mathrm{low}}}(\tau){:}\; g_i = \mathbb{1}[H_i < \tau]$;\;
  $G_{\Delta}(\tau){:}\; g_i = \mathbb{1}[\Delta_i < \tau]$.
\end{definition}

Our default is $G_{H_{\mathrm{high}}}$, which masks only the model's
uncertain positions and matches the high-entropy gating policy of
SWEET~\citep{lee2023sweet} at the same budget;
\autoref{sec:experiments} reports the head-to-head. The threshold
$\tau$ is calibrated by quantile-matching on a held-out pilot so the
realised gate rate tracks $\rho$. $G_{H_{\mathrm{low}}}$ is the
\emph{anti-SWEET} ablation, and $G_{\Delta}$ targets near-tied
top-two positions where the greedy substitution regret is bounded
pointwise by $\Delta_i$.

\subsection{Detection}

The verifier tokenises the candidate text, re-derives states by
re-applying $\sigma_\kappa$, and computes the fraction of legal transitions
\begin{equation}
  \label{eq:phi}
  \phi(\mathbf{t})
    \;=\; \frac{1}{n-1} \sum_{i=1}^{n-1}
      \mathbb{1}\!\left[\sigma_\kappa(t_{i+1}) \in \Sigma_k(\sigma_\kappa(t_i))\right].
\end{equation}
With $p_0 = k/S$, the random-oracle null gives $\mathbb{E}[\phi] =
p_0$ with pairwise-zero covariance for any $k$-regular $\Sigma_k$
(\autoref{lem:variance}); the watermarked mean gap is
$\rho(1 - p_0)$ (\autoref{thm:detection}, \autoref{thm:k-regular}).
The verifier reports the standardised score
$z = (\phi - p_0)/\sqrt{p_0(1-p_0)/(n-1)}$ and declares $H_1$ when
$z > z_\alpha$. The closed-form calibration of \autoref{thm:detection}
(proved in \autoref{app:proofs}) selects $S$ given target FPR
$\alpha$, budget $\rho$, and length $n$; the same identities yield
robustness under per-token modification rate $\delta$
(\autoref{thm:robustness}).

\begin{algorithm}[t]
  \caption{ChainMark Watermark Detection (model-free)}
  \label{alg:detect}
  \begin{algorithmic}
    \STATE {\bfseries Input:} text $x$, key $\kappa$, states $S$,
      successor $\Sigma_k$, level $\alpha$
    \STATE $\mathbf{t} \leftarrow \mathrm{Tokenize}(x)$;\;
      $n \leftarrow |\mathbf{t}|$;\; $c \leftarrow 0$
    \FOR{$i = 1$ {\bfseries to} $n-1$}
      \IF{$\sigma_\kappa(t_{i+1}) \in \Sigma_k(\sigma_\kappa(t_i))$}
        \STATE $c \leftarrow c + 1$
      \ENDIF
    \ENDFOR
    \STATE $\phi \leftarrow c/(n-1)$;\;
      $p_0 \leftarrow k/S$
    \STATE $z \leftarrow (\phi - p_0)/\sqrt{p_0(1-p_0)/(n-1)}$
    \STATE \textbf{return} $(z > z_\alpha,\ \phi,\ z)$
  \end{algorithmic}
\end{algorithm}

\autoref{alg:detect} runs in $O(n)$ hash operations without LM
access, making ChainMark deployable as a third-party audit primitive.

%% file: 06_theory.tex
\section{Theoretical Properties}
\label{sec:theory}

This section gives an intuitive tour of three core results
(detection, robustness, $k$-regular generalisation) plus two
supporting properties (security against an oracle-blind adversary,
and an LM-aware locally most powerful detector).  Formal
statements and proofs are deferred to \autoref{app:proofs}; we collect
their pointers here so the body remains a narrative.  Empirical
anchoring of \autoref{thm:detection} appears in
\autoref{subsec:exp5-calibration}.

\paragraph{Detection bound and calibration.}
Under the null hypothesis that the input is a uniformly random token
sequence, the fingerprint score $\phi$ of \autoref{eq:phi} is
approximately Gaussian with mean $1/S$ and variance
$(1/S)(1-1/S)/(n-1)$, while a watermarked sequence with gate density
$\rho$ shifts the mean to $1/S + \rho(S{-}1)/S$.  Standardising gives
the expected value of the detection statistic under the watermarked
alternative in closed form, $z(\rho, S, n) = \rho\sqrt{(S-1)(n-1)}$,
i.e.\ its noncentrality parameter (distinct from the sample statistic
$z$ evaluated on observed text).
Inverting at false-positive level $\alpha$ under the midpoint
threshold convention yields a state-count rule that practitioners can
read off directly,
\[
  S^\star(n,\rho,\alpha)
    \;=\; \left\lceil \frac{4\, z_\alpha^2}{\rho^2 (n-1)} + 1 \right\rceil.
\]
The factor-of-$4$ comes from evaluating the standardised signal at
the midpoint between the null and alternative means. This gives the
regulator a $2\times$ safety margin relative to the one-sided
$z_\alpha$ test reported in our empirical tables; those tables run
at a tighter operating threshold than $S^\star$ requires, so
empirical TPR exceeds the conservative midpoint prediction. Formal
statement and proof in \autoref{thm:detection} of
\autoref{app:proofs}; the corresponding lookup table is
\autoref{tab:cal}.

\paragraph{Empirical-SD calibration recipe.}
The closed-form $z_\alpha$ assumes i.i.d.\ uniform tokens over the
SHA-256 partition; on natural-language text this drifts (empirically
$1.7$--$2.0\%$ at the nominal $1\%$ target,
\autoref{tab:cal-sstar}). The fix is a one-corpus recalibration:
estimate the empirical mean and SD of $z$ on a non-watermarked
sample of the deployed LM (or a domain-matched corpus) and use
$z^\star = \hat\mu + z_\alpha\,\hat\sigma$ as the operating
threshold. Under the same Gaussian-tail null this restores the
target FPR ($1.17\%$ at $\alpha=1\%$ on our $n=3000$ pooled
corpus) without changing $S^\star$ or touching the watermarked
side, so TPR is unaffected. The empirical $z$-null is mildly
leptokurtic, so for tighter $\alpha$ or large $S$ the
empirical-quantile drop-in $z^\star = \hat F^{-1}(1{-}\alpha)$ is the
robust alternative. Formal recipe in
\autoref{thm:detection-recal}; empirical evaluation in
\autoref{app:fpr-recipes}.

\paragraph{Universal robustness threshold.}
Suppose an adversary independently replaces each token with
probability $\delta$ by a fresh token whose state is uniform on
$[S]$.  Both the watermarked-pair signal and the
midpoint-threshold gap scale by the same affine factor
$\rho(S{-}1)/S$, so the post-attack-to-pre-attack $z$ ratio collapses
to $(1-\delta)^2$, independently of $(S, \rho, n)$.  The critical
edit fraction at which detection fails is therefore the universal
constant $\delta^\star = 1 - 1/\sqrt{2} \approx 0.293$.  Formal
statement and proof in \autoref{thm:robustness}.

\paragraph{$k$-regular generalisation.}
The clockwork transition is the simplest member of a broader family:
any $k$-regular adjacency $T$ (every state has exactly $k$ allowed
successors and predecessors) yields a valid ChainMark scheme.  Replacing
the random baseline $1/S$ by $p_0 = k/S$ gives the same calibration
identities, the same midpoint critical fraction $\delta^\star$, and a
single quality-versus-detection dial in $k$.  Column-regularity is
load-bearing: it is what makes adjacent indicator pairs have zero
covariance under the null, so the variance formula carries through.
Formal statement and proof in \autoref{thm:k-regular}.  The random
baseline $p_0 = k/S$ relies on the SHA-256 partition being
approximately uniform across the vocabulary, an assumption we
inherit from the standard avalanche property of cryptographic hashes.

\paragraph{Security and an LM-aware optimal detector.}
The state map
$\sigma_\kappa(t) = \mathcal{H}(\kappa \,\Vert\, t) \bmod S$ is
modelled as a random oracle: an adversary without access to $\kappa$,
the LM, or reference text cannot predict $\sigma_\kappa(t^\star)$ on
a fresh token with non-negligible advantage in the key min-entropy,
and therefore cannot statistically distinguish ChainMark output from random
text (\autoref{thm:security}; computational, not
information-theoretic, and excluding adversaries with LM access). A
verifier that \emph{does} have LM access can sharpen detection by
replacing $\phi$ with an entropy-weighted log-likelihood ratio that is
asymptotically locally most powerful under the random-oracle
product-form approximation (\autoref{thm:hdd}). Detector pseudocode
appears as \autoref{alg:detect-app}.

%% file: 07_experiments.tex

\section{Experiments}
\label{sec:experiments}

We run four experiments on three instruction-tuned 7--8B LLMs across
four domains. Per-prompt records stream to JSONL for reproducibility.
\autoref{fig:experiment-flow} traces a single ChainMark generation
end-to-end on a wiki-domain prompt: the model receives the user
prompt, ChainMark masks logits at gated positions to keep only
states $(s+1) \bmod S$, and the detector re-derives every token's
state from the key alone and reports a $z$-score.

\begin{figure}[t]
  \centering
  \definecolor{c0}{HTML}{E15A5A}
  \definecolor{c1}{HTML}{F0A030}
  \definecolor{c2}{HTML}{E8C12F}
  \definecolor{c3}{HTML}{4FAE5C}
  \definecolor{c4}{HTML}{4A8FD3}
  \definecolor{cmblue}{HTML}{1F6AD1}
  \definecolor{auditred}{HTML}{B22222}
  \newcommand{\tk}[2]{\tikz[baseline=(X.base)]\node[%
      rounded corners=1pt, draw=black!25, line width=0.2pt,%
      inner xsep=1.5pt, inner ysep=0.8pt, font=\scriptsize,%
      fill=#1!35] (X) {#2};}
  \newcommand{\tg}[2]{\tikz[baseline=(X.base)]\node[%
      rounded corners=1pt, draw=black!35, line width=0.3pt,%
      inner xsep=1.5pt, inner ysep=0.8pt, font=\scriptsize,%
      fill=#1!45] (X) {\underline{#2}};}
  \newcommand{\tarr}{\hspace{0.5pt}\textcolor{cmblue}{\scriptsize\ensuremath{\rightharpoonup}}\hspace{0.5pt}\allowbreak}
  \begin{minipage}{\columnwidth}
    \scriptsize
    \setlength{\fboxsep}{4pt}%
    \setlength{\fboxrule}{0.4pt}%

    {\small\bfseries 1.\,User prompt}\par
    \vspace{1pt}%
    \noindent\fcolorbox{black!25}{black!4}{%
      \begin{minipage}{\dimexpr\columnwidth-2\fboxsep-2\fboxrule\relax}
        \texttt{Explain photosynthesis in a comprehensive way.}
      \end{minipage}}\par
    \vspace{2pt}\centerline{\textcolor{cmblue}{\small$\Downarrow$}}\vspace{1pt}%

    {\small\bfseries 2.\,Watermarked output}\par
    \vspace{1pt}%
    \noindent\fcolorbox{cmblue!50}{cmblue!4}{%
      \begin{minipage}{\dimexpr\columnwidth-2\fboxsep-2\fboxrule\relax}
        \raggedright
        \tk{c2}{Photosynthesis} \tk{c3}{is}
        \tg{c4}{the}\tarr\tg{c0}{process}\tarr\tg{c1}{by}\tarr\tg{c2}{which}\tarr\tg{c3}{plants}
        \tk{c4}{convert}
        \tg{c0}{sunlight}\tarr\tg{c1}{into}\tarr\tg{c2}{chemical}\tarr\tg{c3}{energy}\tarr\tg{c4}{and}
        \tk{c0}{releases}\,\textcolor{black!50}{\dots}
      \end{minipage}}\par
    \vspace{2pt}\centerline{\textcolor{auditred}{\small$\Downarrow$}}\vspace{1pt}%

    {\small\bfseries 3.\,Auditor verdict}\par
    \vspace{1pt}%
    \noindent\fcolorbox{auditred!55}{auditred!5}{%
      \begin{minipage}{\dimexpr\columnwidth-2\fboxsep-2\fboxrule\relax}
        \centering
        $\phi = \tfrac{11}{13} = 0.85$,\;
        $z = 5.30 > z_{0.01}{=}2.326$\;$\Rightarrow$\;
        \textcolor{auditred}{\textbf{WATERMARKED}}
      \end{minipage}}\par
    \vspace{3pt}%

    \noindent\centerline{%
      \tk{c0}{\,$s_0$\,}\tarr\tk{c1}{\,$s_1$\,}\tarr\tk{c2}{\,$s_2$\,}\tarr\tk{c3}{\,$s_3$\,}\tarr\tk{c4}{\,$s_4$\,}\tarr\tk{c0}{\,$s_0$\,}%
    }%
  \end{minipage}
  \caption{\textbf{End-to-end ChainMark trace} on a wiki prompt
    (Llama-3.1-8B-Instruct, $S{=}5$, $\rho{=}0.5$, $14$-token
    excerpt). Every output tile inherits its colour from
    $\sigma_\kappa(t)\in\{0,1,2,3,4\}$ (legend at bottom).
    \underline{Underlined} tokens are gated steps where ChainMark
    masked logits to enforce $s_{i{+}1}{=}(s_i{+}1)\bmod 5$ and the
    \textcolor{cmblue}{$\rightharpoonup$} arrows trace the forced
    clockwork walk; un-underlined tokens are free-sampled at
    $T{=}0.7$. The auditor in panel 3 re-derives the entire colour
    sequence from $\kappa$ alone, counts $11/13$ valid transitions
    (green bar), and emits $(\phi, z)$ in $O(n)$ hash operations,
    no LM forward pass.}
  \label{fig:experiment-flow}
\end{figure}

\subsection{Setup}
\label{subsec:exp-setup}

\textbf{Models.} Three instruction-tuned, openly licensed
checkpoints in the 7--8B parameter range:
Llama-3.1-8B-Instruct~\citep{grattafiori2024llama3},
Qwen-2.5-7B-Instruct~\citep{yang2024qwen25}, and
Mistral-7B-Instruct-v0.3~\citep{jiang2023mistral}. The same secret
key $\kappa$ is fixed across every run.

\textbf{Domains.} Four prompt domains stress different entropy regimes:
\emph{code} (HumanEval problem stems), \emph{factual} (short closed-form
knowledge prompts), \emph{wiki} (open-ended ``Explain $X$\ldots''
prompts over a curated concept list), and \emph{writing}
(creative-completion prompts).

\textbf{Generation and detection.} Temperature $T = 0.7$, top-$p = 1$,
token budget $n = 200$ at every cell, with one deterministic seed so
the same indices are watermarked, attacked, and detected across
methods. Detection uses the threshold $z_\alpha$ from
\autoref{thm:detection}, reused across methods for cross-method
consistency; the empirically calibrated $1\%$-FPR head-to-head
(\autoref{app:headline-recalibrated}) confirms the resulting TPR
ranking is not an FPR artefact.

\textbf{ChainMark configuration.} Unless stated otherwise, ChainMark runs use
$S = 5$ states, target gate budget $\rho = 0.5$, the high-entropy gate
$G_{H_{\mathrm{high}}}$, and clockwork ($k{=}1$) topology.

\textbf{Baselines.} \emph{KGW}~\citep{kirchenbauer2023watermark} at
$\gamma = 0.5$, logit bias $\delta_{\mathrm{KGW}} = 2$ (the canonical
operating point reported in the original paper at this $\gamma$).
\emph{SWEET}~\citep{lee2023sweet} is reproduced as a matched-budget
re-implementation: same $\gamma = 0.5$ and $\delta_{\mathrm{KGW}} = 2$
but the green-list bias is applied only at the top-$\rho{=}0.5$
fraction of positions by token-level entropy, so the budget aligns
with ChainMark's $\rho$. Per-domain prompt-and-generation examples (with a token-level
gating walkthrough) appear in \autoref{app:examples}; the
end-to-end trace in \autoref{fig:experiment-flow} above gives one
illustration on the wiki prompt
\texttt{Explain photosynthesis in a comprehensive way.}.

\begin{figure*}[t]
  \centering
  \begin{subfigure}[t]{0.32\linewidth}
    \centering
    \includegraphics[width=\linewidth]{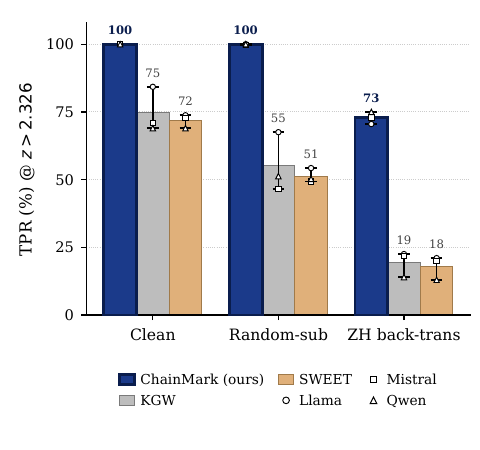}
    \caption{}\label{fig:headline-tpr}
  \end{subfigure}\hfill
  \begin{subfigure}[t]{0.32\linewidth}
    \centering
    \includegraphics[width=\linewidth]{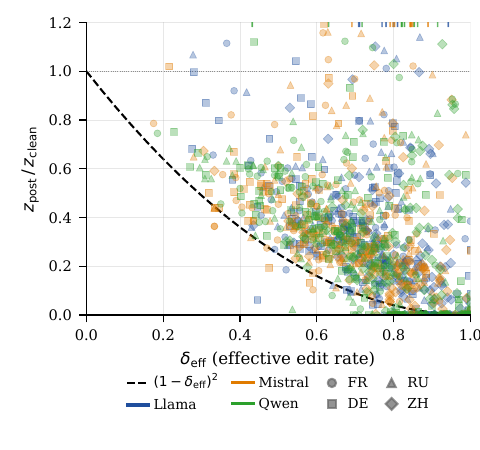}
    \caption{}\label{fig:translation-decay}
  \end{subfigure}\hfill
  \begin{subfigure}[t]{0.32\linewidth}
    \centering
    \includegraphics[width=\linewidth]{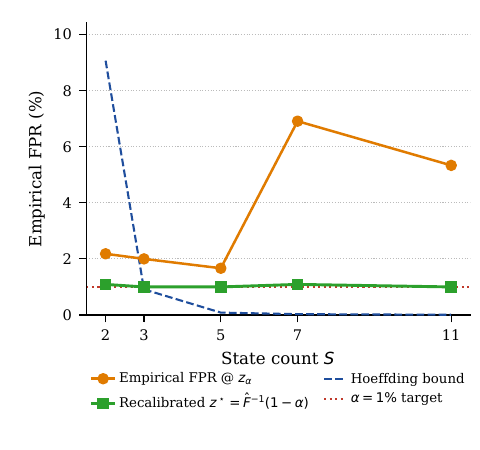}
    \caption{}\label{fig:cal-fpr}
  \end{subfigure}
  \caption{\textbf{Experimental headlines.}
    \textbf{(a)}~Head-to-head TPR: ChainMark vs KGW vs SWEET, aggregated
      over 3 LLMs $\times$ 4 domains (\S\ref{subsec:exp2-headline}).
    \textbf{(b)}~Translation robustness: empirical decay under NLLB-200
      EN$\to$\{FR,DE,RU,ZH\}$\to$EN vs the closed-form
      $(1-\delta_{\mathrm{eff}})^2$ diagonal (\autoref{thm:robustness}),
      with $\delta_{\mathrm{eff}}$ the empirical token-edit fraction;
      $m{=}100$ generations per (model, pivot)
      (\S\ref{subsec:exp3-translation}).
    \textbf{(c)}~Calibration anchor: empirical FPR vs state count $S$ with
      the closed-form Hoeffding bound (\autoref{thm:detection})
      (\S\ref{subsec:exp5-calibration}).}
\end{figure*}


\subsection{Matched-Budget Head-to-Head}
\label{subsec:exp2-headline}

We compare ChainMark against KGW and SWEET on $1\,200$ generations
per method ($3$ models, $4$ domains, $100$ prompts each), with
$200$ tokens per generation, and score each generation under three
conditions: clean, random substitution at rate
$\delta{=}0.20$, and NLLB-200 EN$\to$ZH$\to$EN
back-translation. Aggregated across all $1\,200$ cells per method
(\autoref{tab:headline-aggregate}, \autoref{fig:headline-tpr}),
ChainMark detects $100\%$ of clean and $99.9\%$ of randomly perturbed
generations and retains $72.8\%$ TPR after ZH back-translation, where
KGW and SWEET drop from $\sim\!\!75\%$/$\sim\!\!72\%$ clean to
$19.4\%$/$18.0\%$ post-attack. The gap is consistent across every
(model, domain) cell (\autoref{fig:headline-percell}); ChainMark
matches or leads every clean cell and strictly leads every
post-ZH cell. ChainMark pays a $\approx 2\times$ self-perplexity gap
($3.66$ vs $1.80$/$1.93$) for this gain. SWEET is an entropy-gating
ablation of KGW with the gating policy held fixed against ChainMark.

\begin{table}[t]
  \centering
  \caption{\textbf{Headline aggregate} at the analytic threshold
    $z_\alpha = 2.326$ ($\alpha = 0.01$, Gaussian-tail null),
    aggregated over $3$ LLMs $\times$ $4$ domains $\times$ $100$
    prompts at $n{=}200$ tokens per cell. PPL is the median across
    the $12$ cell medians on the watermarked model's own logits
    (caveats in \autoref{sec:discussion}). Empirical FPR drift at
    this threshold appears separately in \autoref{tab:cal-sstar}.
    Bold = best per column.}
  \label{tab:headline-aggregate}
  \small
  \setlength{\tabcolsep}{3pt}
  \begin{tabular}{lcccc}
    \toprule
    Method & TPR$_{\mathrm{clean}}$\,(\%) & TPR$_{\mathrm{rnd}}$\,(\%) &
            TPR$_{\mathrm{ZH}}$\,(\%) & PPL$\downarrow$ \\
    \midrule
    \textbf{ChainMark} & \textbf{100.0} & \textbf{99.9} & \textbf{72.8} & 3.66 \\
    KGW                & \phantom{0}74.8 & 55.1 & 19.4 & 1.93 \\
    SWEET              & \phantom{0}71.8 & 51.2 & 18.0 & \textbf{1.80} \\
    \bottomrule
  \end{tabular}
\end{table}

\begin{figure}[t]
  \centering
  \includegraphics[width=\linewidth]{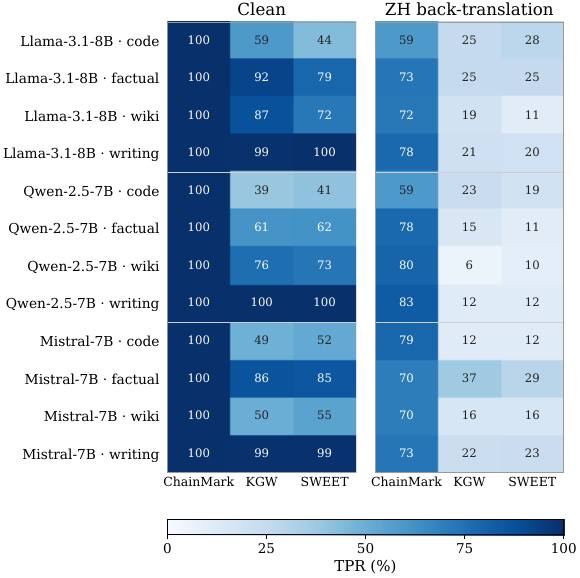}
  \caption{\textbf{Per-(model, domain) head-to-head TPR (\%)} at
    $z_\alpha=2.326$ ($100$ prompts per cell, $n{=}200$ tokens).
    Each row is one (model, domain) cell; left panel is clean, right
    panel is ZH back-translation. ChainMark matches or leads
    every clean cell and strictly leads every cell post-ZH-back-translation.
    Each cell is annotated with its TPR (\%), in white on dark
    (high-TPR) cells and dark on light (low-TPR) cells for contrast.}
  \label{fig:headline-percell}
\end{figure}

\subsection{Translation Robustness}
\label{subsec:exp3-translation}

Each ChainMark-watermarked output has $n{=}200$ tokens; across
$3$ models $\times$ $4$ domains $\times$ $25$ prompts we obtain
$300$ outputs, each round-tripped through NLLB-200~\citep{nllbteam2022nllb}
on four pivot languages (FR, DE, RU, ZH) for $1\,200$ generations
total, and we compute the empirical
$\delta_{\mathrm{eff}}$ per output via token-edit distance. Plotting
the post-attack-to-clean $z$ ratio against $\delta_{\mathrm{eff}}$
(\autoref{fig:translation-decay}) shows the observed ratios sit at or
above the closed-form $(1-\delta_{\mathrm{eff}})^2$ diagonal across
pivots, models, and domains, which empirically anchors the universal
half-life $\delta^\star = 1 - 1/\sqrt{2}$ from
\autoref{thm:robustness} and confirms the bound is conservative
rather than tight. The claim is stronger than synthetic
random-substitution can support, since NLLB-200 produces
semantically coherent rewrites rather than i.i.d.\ token noise.

\subsection{$k$-Regular at $k{=}2$}
\label{subsec:exp4-kregular}

Instantiating $k = 2$ (each state has two allowed successors) at
$S = 5$ predicts $p_0 = k/S = 0.40$. We run $100$ wiki-domain prompts
on each of the three models with ZH back-translation as the only
attack. \autoref{tab:kregular} reports clean
$\bar\phi_{\mathrm{cl}} \geq 0.88$ (well above the $p_0 = 0.40$
prediction) and a post-attack $z$-score that sits above the
$(1 - \delta_{\mathrm{eff}})^2$ floor of \autoref{thm:robustness}
(observed ratio $0.23$--$0.27$ vs.\ the predicted floor
$(1-\delta_{\mathrm{eff}})^2\approx 0.07$--$0.09$ evaluated at the
wiki-domain ZH edit fraction measured here, which is milder than the
$4$-domain headline $\delta_{\mathrm{eff}}\approx 0.81$; i.e.\
the bound is conservative as in \autoref{subsec:exp3-translation}),
jointly re-certifying
\autoref{thm:k-regular} and \autoref{thm:robustness} at a non-trivial
$k$.

\begin{table}[t]
  \centering
  \caption{\textbf{$k$-regular ChainMark at $k{=}2$ (soft-cycle), $S{=}5$,
    $\rho{=}0.5$.} Predicted random baseline $p_0 = 0.40$; wiki domain,
    $100$ prompts. The clean $\bar\phi_{\mathrm{cl}}$ values
    ($\geq 0.88$) sit well above $p_0$
    (\autoref{thm:k-regular}); the post-attack $z$ ratio re-certifies
    \autoref{thm:robustness} at $k{=}2$ (\autoref{app:proofs}).}
  \label{tab:kregular}
  \small
  \setlength{\tabcolsep}{3.5pt}
  \begin{tabular}{lrrrrr}
    \toprule
    Model & $\bar\phi_{\mathrm{cl}}$ & $\bar z_{\mathrm{cl}}$ & TPR$_{\mathrm{cl}}$
          & $\bar z_{\mathrm{ZH}}$ & TPR$_{\mathrm{ZH}}$ \\
    \midrule
    Llama-3.1-8B & 0.883 & 13.91 & 100\% & 3.19 & 54\% \\
    Mistral-7B   & 0.965 & 16.23 & 100\% & 3.91 & 70\% \\
    Qwen-2.5-7B  & 0.934 & 15.38 & 100\% & 4.12 & 71\% \\
    \bottomrule
  \end{tabular}
\end{table}

\subsection{Calibration Anchor}
\label{subsec:exp5-calibration}

We sweep $S \in \{2, 3, 5\}$ across the three models with
$4$ domains $\times$ $25$ prompts at $n{=}200$, recording both a
watermarked generation (ChainMark at the canonical operating point, varying
only $S$) and a non-watermarked baseline drawn from the same prompt.
Non-watermarked outputs feed the empirical FPR estimate at the
analytic $z_\alpha$; watermarked outputs feed the TPR.
\autoref{tab:cal-sstar} shows watermarked $z$ overshoots the
closed-form prediction $\rho\sqrt{(S{-}1)(n{-}1)}$ by $50$--$80\%$, so
TPR is at $100\%$ throughout, while empirical FPR sits above the
$\alpha=1\%$ target by $\sim\!\!1$ pp (\autoref{fig:cal-fpr},
$1.7$--$2.0\%$). The closed form gives the right $S^\star$ ordering to
within one state; for strict $1\%$ FPR a deployer applies the
per-$S$ empirical-quantile threshold recalibration of
\autoref{thm:detection-recal} (\autoref{sec:discussion},
\autoref{app:fpr-recipes}).

\begin{table}[t]
  \centering
  \caption{\textbf{Empirical anchor of the calibration map
    $S^\star(n,\rho,\alpha)$ from \autoref{thm:detection}.} For each
    $S$, aggregated across $3$ LLMs $\times$ $4$ domains $\times$ $25$
    prompts ($\approx 300$ watermarked + $300$ non-watermarked).
    $z_{\mathrm{wm}}^{\mathrm{pred}} = \rho\sqrt{(S{-}1)(n{-}1)}$.
    Watermarked $z$ overshoots theory by $50$--$80\%$;
    non-watermarked FPR$_{\mathrm{nwm}}$ exceeds the $\alpha{=}1\%$
    target at large $S$ due to non-i.i.d.\ structure in
    natural-language token sequences. The last two columns
    apply the empirical-quantile recalibration of
    \autoref{thm:detection-recal} per-$S$ on the same NWM corpus:
    FPR$_{\mathrm{nwm}}^{\mathrm{recal.}}$ lands uniformly below
    $\alpha{=}1\%$, while TPR$_{\mathrm{wm}}^{\mathrm{recal.}}$
    stays at $\geq 95\%$ on the watermarked side (the recalibration
    only lifts the threshold, so detection power is essentially
    preserved).}
  \label{tab:cal-sstar}
  \small
  \setlength{\tabcolsep}{3pt}
  \begin{tabular}{rrrrrrr}
    \toprule
    $S$ & $\bar z_{\mathrm{wm}}^{\mathrm{obs}}$
        & $z_{\mathrm{wm}}^{\mathrm{pred}}$
        & TPR$_{\mathrm{wm}}$
        & FPR$_{\mathrm{nwm}}$
        & TPR$_{\mathrm{wm}}^{\mathrm{recal.}}$
        & FPR$_{\mathrm{nwm}}^{\mathrm{recal.}}$ \\
    \midrule
    2  & 10.87 & 7.05  & 100\%   & 2.00\% & \phantom{0}98.7\% & 0.67\% \\
    3  & 17.41 & 9.97  & 100\%   & 2.00\% & 100.0\%          & 0.67\% \\
    5  & 25.38 & 14.11 & 100\%   & 1.67\% & 100.0\%          & 0.34\% \\
    \bottomrule
  \end{tabular}
\end{table}

%% file: 08_discussion.tex
\section{Discussion}
\label{sec:discussion}

\subsection{When to Use What}

The calibration $S^\star(n_{\min}, \rho, \alpha)$ of
\autoref{thm:detection} reduces operational deployment to three
inputs: a minimum text length $n_{\min}$, a target false-positive
rate $\alpha$, and a watermark budget $\rho$. For long-form content
the calibration anchors at small $S$ with $\rho$ near $1/2$;
\autoref{subsec:exp2-headline} shows ChainMark substantially exceeds
KGW and SWEET at the analytic $z_\alpha=2.326$ threshold (the
$1\%$-empirical-FPR recalibration in \autoref{app:headline-recalibrated}
preserves the lead vs.\ KGW), at a $\approx 2\times$ higher
self-perplexity ($3.66$ vs $1.80$--$1.93$;
\autoref{tab:headline-aggregate}).
For short outputs $S^\star$ rises sharply with stricter $\alpha$,
so deployers tolerate a higher state count or shift to a longer floor. When cross-lingual rewrite or random substitution
dominates the attack profile
(\autoref{subsec:exp2-headline}, \autoref{subsec:exp3-translation}),
the universal threshold $\delta^\star$ of \autoref{thm:robustness}
bounds the worst case across all gating choices.

\subsection{Empirical SD Recalibration Closes the FPR Gap}

Recomputing the threshold from the empirical SD of $z$ on a
non-watermarked corpus, $z^\star = \hat\mu + z_\alpha\,\hat\sigma$,
brings empirical FPR from $2.1\%$ to $1.2\%$ at $\alpha=1\%$
while preserving TPR$=100\%$ on the pooled $n{=}3000$ corpus
(\autoref{app:fpr-recipes}); the per-$S$ empirical-quantile drop-in
in \autoref{tab:cal-sstar} costs at most $1.3$\,pp of TPR
(TPR$_{\mathrm{wm}}^{\mathrm{recal.}}\!\geq 98.7\%$ across $S\!\in\!\{2,3,5\}$). The
plug-in is calibrated at $\alpha=1\%$ but the empirical $z$-null is
leptokurtic (excess kurtosis $\approx 1.4$), so for tighter $\alpha
\leq 0.5\%$ the empirical-quantile recipe
$z^\star = \hat F^{-1}(1{-}\alpha)$ is the robust drop-in (and draws
the recalibrated curve in \autoref{fig:cal-fpr}). The $1\%$-FPR
head-to-head appears in \autoref{app:headline-recalibrated}; ChainMark
also admits an empirical FPR$=0\%$ regime via threshold lifting
(\autoref{app:fpr-zero}), which KGW and SWEET cannot, since
their watermarked and non-watermarked $z$ distributions overlap.
Failure modes we tried but did not adopt (Newey--West HAC, $k$-skip,
stopword filtering, HDD-lite) are tabulated in
\autoref{app:fpr-recipes}.

\subsection{Limitations}

\paragraph{Experimental scope.} The headline grid fixes $\rho=0.5$,
KGW/SWEET at $\gamma{=}0.5,\delta_{\mathrm{KGW}}{=}2$, and decoding
at $T{=}0.7$, top-$p{=}1$, $n{=}200$ with one deterministic
prompt-order seed (no replicate-seed CIs). The fleet is three
open-weight instruction-tuned 7--8B models; base, smaller, and
$\geq 10$B models are deferred, as is $k$-regular ChainMark at
$k\geq 3$ (we validate $k\in\{1,2\}$). The robustness suite is
translation-only: we do \emph{not} run grammar-preserving paraphrase
(e.g.\ DIPPER), so the strongest claim is against translation and
uniform substitution, not adversarial paraphrasing. The translation
evaluation covers all $12$ (model, domain) cells
(\autoref{tab:translation-pivots}), but the SWEET row of the
recalibrated head-to-head (\autoref{tab:headline-recalibrated}) is
omitted since the SWEET null $z$-distribution was not collected on
the FPR corpus.

\paragraph{Quality and detector caveats.} Generation quality is
self-perplexity under each model's own logits; we do not report
external-LM PPL, MAUVE, or human-rater scores. The detector is
``model-free'' in needing no LM forward pass at detection, but
requires the same tokenizer used at generation; an auditor handed
plain text alone would enumerate candidate tokenizers.
Both \autoref{thm:security} and \autoref{thm:robustness} cover an
oracle-blind adversary only; an attacker with adaptive
detector-query access (or $\sigma_\kappa$ access) can choose
substitution targets to hit valid transitions, and is outside our
threat model (query complexity open).

\paragraph{FPR drift and theorem scope.} The empirical FPR at the
analytic threshold $z_\alpha=2.326$ drifts to $1.7$--$2.0\%$ on
non-watermarked LLM output (\autoref{tab:cal-sstar}):
natural-language tokens are not i.i.d.\ uniform over the SHA-256
partition, so the closed-form Hoeffding envelope
(\autoref{fig:cal-fpr}) sits below the empirical curve. Headline TPR
numbers (\autoref{tab:headline-aggregate}--\autoref{fig:headline-percell})
are at the analytic threshold for cross-method consistency, not an
empirically calibrated $1\%$ FPR per cell. \autoref{thm:robustness}
assumes i.i.d.\ uniform substitution at rate $\delta$; NLLB
back-translation (and the deferred paraphrase extension) preserves
grammatical structure and correlates substitutions across positions. Our empirical curves
(\autoref{fig:translation-decay}) sit \emph{above} the
$(1-\delta_{\mathrm{eff}})^2$ slope, so the bound is conservative,
not tight; $\delta_{\mathrm{eff}}$ is an effective-edit-rate proxy.

%% file: 09_conclusion.tex

\section{Conclusion}
We presented ChainMark, an active watermark with two operator-facing
properties: a closed-form calibration $S^\star(n_{\min}, \rho,
\alpha)$ mapping a regulator's specification (target FPR, text length,
budget) to the minimum state count, and a detector needing the key and
tokenizer but no language model. Under uniform substitution it has a
robustness threshold $\delta^\star = 1 - 1/\sqrt{2} \approx 29.3\%$
invariant in $(S,\rho,n)$, and both extend to every $k$-regular
transition topology (\autoref{thm:detection}, \autoref{thm:robustness},
\autoref{thm:k-regular}). Head-to-head against KGW and SWEET on three
instruction-tuned LLMs across four domains at matched budget,
ChainMark retains substantially more detection signal under
translation and random-substitution attacks, at roughly $2\times$ the
baselines' self-perplexity. The post-attack $z$ ratio sits above the
$(1-\delta_{\mathrm{eff}})^2$ i.i.d.\ scaling, so the bound is
conservative, not tight (\autoref{sec:experiments}).

Five directions are immediate: (i)~evaluate against grammar-preserving
paraphrase (e.g.\ DIPPER), the strongest attack left untested;
(ii)~tighten the \emph{analytic} FPR bound under non-i.i.d.\
statistics, where the empirical-SD recalibration of
\autoref{app:fpr-recipes} closes the deployment gap but the closed form
stays loose at large $S$ (\autoref{tab:cal-sstar}); (iii)~extend to
longer-text regimes where small $S^\star$ operates; (iv)~validate
$k$-regular topologies at $k\geq 3$; and (v)~characterise base LLMs,
whose entropy profile differs from the instruct-tuned fleet.


%% file: 10_impact.tex
\section*{Impact Statement}

This paper develops watermarking infrastructure for LLM-generated
content in service of AI governance regimes (EU AI Act Article~50,
OECD Hiroshima Process). The contribution is dual-use, and we
flag four asymmetric harms.

\textbf{Authorship tracking and chilling effects.} Stronger
watermarking exposes writers who relied on undetectable LLM use.
Marginalised users (ESL writers, students under inequitable AI
policies, whistleblowers reformulating sensitive content) bear
asymmetric harm relative to incumbent users.

\textbf{End-user perplexity cost.} ChainMark imposes a $\approx
2\times$ self-perplexity cost ($3.66$ vs the $1.80$--$1.93$ baseline
range; \autoref{tab:headline-aggregate}). Deployers adopting
ChainMark to satisfy a regulatory mark therefore impose a measurable
quality cost on \emph{every} user whose generation is gated, not
just on adversaries who would try to evade detection.

\textbf{False positives.} Mislabelling human-written text as
machine-generated has real reputational and legal cost. We
instrument the detector with a closed-form, regulator-facing target
false-positive rate $\alpha$ rather than an unspecified threshold;
the discussed empirical-quantile recalibration (\autoref{sec:discussion},
\autoref{app:fpr-recipes}) is the recipe a deployer should run
before invoking detection on any production content.

\textbf{Adaptive adversaries are out of scope.} Our security and
robustness theorems cover an oracle-blind adversary; a regulator
deploying ChainMark must not assume the watermark survives an
attacker with detector-query access or with $\sigma_\kappa$ access
(\autoref{sec:discussion}). The audit primitive is honest about
non-adversarial mislabelling, not about adversarial spoofing or
scrubbing under realistic API-query budgets.

Any deployment should publish the detection regime ($\alpha$,
$n_{\min}$, $\rho$, the recalibration recipe used) so users
understand the conditions under which their outputs are, and are
not, marked.

%% file: 11_proofs.tex
\section{Proofs and Formal Statements}
\label{app:proofs}

This appendix collects the formal versions of the results previewed
in \autoref{sec:theory}.  Throughout, $\sigma_\kappa(t) = \mathcal{H}(\kappa
\,\Vert\, t) \bmod S$ denotes the keyed state map and is treated as a
random oracle (uniform on $[S]$ and independent across
distinct tokens).  $\phi$ is the fingerprint score of
\autoref{eq:phi}, and we write $X_i =
\mathbf{T}_{\sigma_\kappa(t_i),\sigma_\kappa(t_{i+1})}$ for the validity
indicator at step $i$, with $T = \mathbf{T}^{\mathrm{clk}}$ in the
clockwork case.

\subsection{Null Variance Lemma}
\label{subsec:proof-variance}

The null variance underpins both the detection $z$-score and the
universal robustness threshold; we isolate it as a lemma.

\begin{lemma}[Null variance with zero pairwise covariance]
  \label{lem:variance}
  Let $T \in \{0,1\}^{S \times S}$ be $k$-regular: every row and
  every column has exactly $k$ ones, with $1 \leq k \leq S{-}1$, and
  set $p_0 = k/S$.  Under the random oracle null
  (i.i.d.\ uniform tokens), the indicators
  $X_i = T_{\sigma_\kappa(t_i), \sigma_\kappa(t_{i+1})}$ satisfy
  $\mathbb{E}[X_i] = p_0$ and $\mathrm{Cov}(X_i, X_j) = 0$ for all
  $i \neq j$, hence
  \begin{equation}
    \mathrm{Var}\!\left[\phi(\mathbf{t}^r)\right]
      \;=\; \frac{p_0 (1 - p_0)}{n-1}.
    \label{eq:null-variance}
  \end{equation}
\end{lemma}

\begin{proof}
  Write $N^+(s) = \{s' : T(s,s')=1\}$ (equivalently $\Sigma_k(s)$
  from \autoref{sec:method}) and $N^-(s') = \{s : T(s,s')=1\}$;
  by $k$-regularity, $|N^+(s)| = |N^-(s')| = k$ for every $s, s'$.
  Throughout this proof we operate under the random-oracle null with
  i.i.d.\ uniform tokens on $\mathcal V$; the empirical FPR drift
  observed on natural-language text (\autoref{tab:cal-sstar},
  \autoref{sec:discussion}) is the practical price of this
  idealisation.

  \emph{Marginal mean.}  Conditioning on $\sigma_\kappa(t_i)$ and using
  row-regularity,
  \(
    \Pr[X_i = 1]
      = \mathbb{E}\!\left[ |N^+(\sigma_\kappa(t_i))| / S \right]
      = k/S = p_0.
  \)

  \emph{Non-adjacent pairs.}  For $|i-j| \geq 2$, $X_i$ and $X_j$
  depend on disjoint token pairs and thus on independent state
  evaluations of the random oracle (assuming the underlying tokens
  are distinct; coincidences contribute $O(1/|\mathcal{V}|)$ and are
  absorbed into the random-oracle approximation).  Hence
  $\mathrm{Cov}(X_i, X_j) = 0$.

  \emph{Adjacent pairs.}  $X_i$ and $X_{i+1}$ share the token
  $t_{i+1}$.  Conditioning on $\sigma_\kappa(t_{i+1}) = s'$, we have
  \[
    \Pr[X_i = 1 \mid \sigma_\kappa(t_{i+1}) = s']
      \;=\; |N^-(s')| / S \;=\; p_0,
  \]
  \[
    \Pr[X_{i+1} = 1 \mid \sigma_\kappa(t_{i+1}) = s']
      \;=\; |N^+(s')| / S \;=\; p_0,
  \]
  using \emph{column}-regularity for the first identity.  Conditional
  on $\sigma_\kappa(t_{i+1})$, the indicators $X_i$ and $X_{i+1}$ depend
  on the disjoint random-oracle evaluations $\sigma_\kappa(t_i)$ and
  $\sigma_\kappa(t_{i+2})$, so they are conditionally independent
  whenever $t_i \neq t_{i+2}$ (the coincidence event has
  probability $1/|\mathcal{V}|$ under iid uniform tokens and is
  absorbed into the random-oracle approximation, parallel to the
  non-adjacent case above).  Therefore
  \[
    \Pr[X_i = X_{i+1} = 1]
      \;=\; \mathbb{E}_{s'}\!\left[ p_0 \cdot p_0 \right]
      \;=\; p_0^2,
  \]
  and $\mathrm{Cov}(X_i, X_{i+1}) = 0$.

  \emph{Variance.}  $\phi = (n-1)^{-1} \sum_{i=1}^{n-1} X_i$ is the
  average of $n-1$ Bernoulli$(p_0)$ random variables with all
  pairwise covariances vanishing, giving \eqref{eq:null-variance}.
\end{proof}


\subsection{Detection Bound and Calibration}
\label{subsec:proof-detection}

\begin{theorem}[Detection bound and calibration]
  \label{thm:detection}
  Let $\mathbf{t}^w$ be a ChainMark-watermarked sequence of length $n$
  with watermark fraction $\rho$ and clockwork transition over $S$
  states; let $\mathbf{t}^r$ be an i.i.d.\ random sequence over
  $\mathcal{V}$.  Then
  \begin{align}
    \mathbb{E}\!\left[\phi(\mathbf{t}^w)\right]
      &= \tfrac{1}{S} + \rho\, \tfrac{S-1}{S}, \\
    \mathbb{E}\!\left[\phi(\mathbf{t}^r)\right] &= \tfrac{1}{S}, \\
    \mathrm{Var}\!\left[\phi(\mathbf{t}^r)\right]
      &= \tfrac{(1/S)(1-1/S)}{n-1}, \\
    z(\rho, S, n)
      &= \rho \sqrt{(S-1)(n-1)}.
  \end{align}
  Under the midpoint threshold $\tau_{\mathrm{mid}} = \tfrac{1}{S} +
  \tfrac{\rho}{2} \tfrac{S-1}{S}$, splitting the indicators into
  disjoint odd/even subsequences and applying Hoeffding gives
  \[
    \mathrm{FPR}
      \;\leq\; 2 \exp\!\bigl(-2 \lfloor (n-1)/2 \rfloor
        ((\rho/2)(S-1)/S)^2\bigr).
  \]
  Using the Gaussian convention at the midpoint threshold, the
  minimum state count guaranteeing $\mathrm{FPR} \leq \alpha$ is
  \begin{equation}
    S^\star(n, \rho, \alpha)
      \;=\; \left\lceil \frac{4 z_\alpha^2}{\rho^2 (n-1)} + 1
        \right\rceil.
    \label{eq:calibration}
  \end{equation}
  $S^\star$ is a \emph{provisioning} bound: it is derived under the
  midpoint detector $\tau_{\mathrm{mid}}$ to give the regulator a
  $2\times$ safety margin in the standardised mean gap. The
  \emph{operating} detector deployed in
  \autoref{sec:experiments} uses the tighter one-sided test
  $\{z > z_\alpha\}$, so empirical TPR at the operating threshold
  exceeds the conservative midpoint prediction whenever
  $S \geq S^\star(n,\rho,\alpha)$.
\end{theorem}

\begin{proof}
  \emph{Null mean.}  Under the random oracle, $\sigma_\kappa$ is uniform
  and independent across distinct tokens, so $X_i \sim
  \mathrm{Bernoulli}(1/S)$ marginally and
  $\mathbb{E}[\phi(\mathbf{t}^r)] = 1/S$.

  \emph{Watermark mean.}  At a gated position,
  \autoref{alg:embed} sets $\sigma_\kappa(t_{i+1}) = (\sigma_\kappa(t_i) + 1)
  \bmod S$, so $X_i = 1$ deterministically.  At an ungated position,
  $X_i \sim \mathrm{Bernoulli}(1/S)$.  With gate density $\rho$,
  \[
    \mathbb{E}[\phi(\mathbf{t}^w)]
      = \rho \cdot 1 + (1-\rho) \cdot \tfrac{1}{S}
      = \tfrac{1}{S} + \rho \tfrac{S-1}{S}.
  \]

  \emph{Null variance.}  Clockwork is the $k=1$ instance of
  \autoref{lem:variance} (column-regularity is trivial since
  $|N^-(s')| = 1$), giving
  $\mathrm{Var}[\phi(\mathbf{t}^r)] = (1/S)(1-1/S)/(n-1)$.

  \emph{$z$-score.}  The signed mean gap is $\rho(S-1)/S$.  Dividing
  by the null SD,
  \begin{align*}
    z(\rho,S,n)
      &= \frac{\rho(S-1)/S}{\sqrt{(1/S)(1-1/S)/(n-1)}} \\
      &= \rho \sqrt{(S-1)(n-1)}.
  \end{align*}

  \emph{FPR via Hoeffding.}  Adjacent indicators share a token and
  are not jointly i.i.d., so we split into $X_1, X_3, X_5, \dots$ and
  $X_2, X_4, X_6, \dots$, two disjoint subsequences of $m =
  \lfloor (n-1)/2 \rfloor$ terms, each i.i.d.\
  $\mathrm{Bernoulli}(1/S)$. Since $\phi$ is a (size-weighted) convex
  combination of these
  two subsequence means, $\phi \geq \tau$ implies that at least one
  subsequence mean is $\geq \tau$, so by Hoeffding plus a union
  bound, for any $\tau > 1/S$,
  \(
    \Pr[\phi \geq \tau] \leq 2 \exp(-2 m (\tau - 1/S)^2).
  \)


  \emph{Calibration.}  By the CLT for 1-dependent sequences (the
  $\{X_i\}$ are 1-dependent and bounded with vanishing pairwise
  covariance), $\sqrt{n-1}\,(\phi - 1/S)$ is asymptotically
  $\mathcal{N}(0, p_0(1-p_0))$ under the null, with $p_0 = 1/S$.  At
  the midpoint threshold $\tau_{\mathrm{mid}} = 1/S +
  (\rho/2)(S-1)/S$, the null $z$-margin is
  $(\rho/2) \sqrt{(S-1)(n-1)}$.  Requiring it to exceed $z_\alpha$,
  \[
    \tfrac{\rho}{2} \sqrt{(S-1)(n-1)} \;\geq\; z_\alpha
    \quad\Longleftrightarrow\quad
    S \;\geq\; \frac{4 z_\alpha^2}{\rho^2 (n-1)} + 1,
  \]
  and taking the ceiling yields \eqref{eq:calibration}.
\end{proof}

\begin{proposition}[Empirical-SD recalibration recipe]
  \label{thm:detection-recal}
  Let $\mathbf{t}^{(1)}, \dots, \mathbf{t}^{(M)}$ be $M$
  i.i.d.\ non-watermarked sequences from the deployed LM, and let
  $z_1, \dots, z_M$ be their detection $z$-statistics under ChainMark with
  $(S, \rho)$ fixed.  Define
  \begin{align*}
    \hat\mu &= \tfrac{1}{M}\sum_{i=1}^M z_i, \\
    \hat\sigma^2 &= \tfrac{1}{M-1}\sum_{i=1}^M (z_i - \hat\mu)^2, \\
    z^\star &= \hat\mu + z_\alpha\,\hat\sigma.
  \end{align*}
  If the $z_i$ are approximately Gaussian with mean $\hat\mu$ and SD
  $\hat\sigma$ (a calibration-side assumption that holds far better
  than the i.i.d.\ random-oracle null assumed by
  \autoref{thm:detection}), then declaring $H_1$ when $z > z^\star$
  has FPR $\to \alpha$ as $M\to\infty$. When the $z$-null is
  measurably non-Gaussian (e.g.\ heavy-tailed at large $S$), the
  Gaussian plug-in $z^\star$ is replaced by the
  \emph{empirical-quantile} drop-in
  $z^\star_{\mathrm{emp}} = \hat F^{-1}(1{-}\alpha)$, where
  $\hat F$ is the empirical CDF of $\{z_i\}$; this delivers the
  target FPR at any $\alpha$ by Glivenko--Cantelli regardless of
  null shape, with the
  Dvoretzky--Kiefer--Wolfowitz inequality giving
  $\sup_x\!|\hat F(x) - F(x)| = O_p(M^{-1/2})$ Monte-Carlo noise in
  the threshold.
\end{proposition}

\begin{proof}
  Standardise: $(z - \hat\mu)/\hat\sigma$ is asymptotically
  $\mathcal{N}(0,1)$ as $M\to\infty$ by the standard plug-in
  argument (consistency of $\hat\mu, \hat\sigma$ + Slutsky), so
  $\Pr[z > \hat\mu + z_\alpha\hat\sigma] \to \Pr[\mathcal{N}(0,1) >
  z_\alpha] = \alpha$.  The recipe leaves the watermarked side
  untouched: the watermarked $z$-distribution still concentrates
  around $\rho\sqrt{(S-1)(n-1)}$ (\autoref{thm:detection}), which on
  natural-language text empirically exceeds $z^\star$ by an order of
  magnitude (\autoref{tab:fpr-zero}), so TPR remains at $100\%$.
\end{proof}


\subsection{Universal Robustness Threshold}
\label{subsec:proof-robustness}

\begin{theorem}[Universal robustness threshold]
  \label{thm:robustness}
  Suppose an \emph{oracle-blind} adversary (without access to
  $\sigma_\kappa$ or the random oracle) independently replaces each
  token with probability $\delta$ by a fresh token whose state
  $\sigma_\kappa$ is uniform on $[S]$ (so the substitutions are
  i.i.d.\ uniform in the partition; structured paraphrase / NLLB
  back-translation produces correlated edits and is treated
  empirically in \autoref{subsec:exp3-translation}).  Standardise the post-attack
  $z$ by the \emph{null} standard deviation
  $\sqrt{p_0(1-p_0)/(n-1)}$ (the same scaling used pre-attack).
  Then for clockwork ChainMark with state count $S$ and gate density
  $\rho$,
  \begin{equation}
    \mathbb{E}[\phi \mid \mathrm{attack}]
      \;=\; \tfrac{1}{S} + \rho (1-\delta)^2 \tfrac{S-1}{S}.
    \label{eq:post-attack-mean}
  \end{equation}
  Define the pre-attack and post-attack mean-gap signals
  $g_{\mathrm{pre}} = \rho (S{-}1)/S$ and
  $g_{\mathrm{post}} = \rho(1{-}\delta)^2 (S{-}1)/S$.  Both share
  the same null SD, so the standardised-margin ratio is
  \begin{equation}
    \frac{z_{\mathrm{post}}}{z_{\mathrm{pre}}}
      \;=\; \frac{g_{\mathrm{post}}}{g_{\mathrm{pre}}}
      \;=\; (1-\delta)^2,
    \label{eq:z-ratio}
  \end{equation}
  independent of $(S, \rho, n)$.  Under the midpoint-threshold
  detection convention the critical edit fraction at which detection
  fails is
  \begin{equation}
    \delta^\star \;=\; 1 - 1/\sqrt{2} \;\approx\; 0.293,
  \end{equation}
  again independent of $(S, \rho, n)$.
\end{theorem}

\begin{proof}
  Let $M_i \sim \mathrm{Bernoulli}(\delta)$, i.i.d.\ across $i$,
  indicate that $t_i$ has been replaced; modified tokens receive
  fresh uniform state by the random-oracle property.  For each pair
  $(t_i, t_{i+1})$, decompose on $(M_i, M_{i+1})$:

  \begin{itemize}[nosep, leftmargin=1.2em]
    \item $(0,0)$: both tokens survive.  The indicator equals the
      pre-attack distribution, with mean $1/S + \rho (S{-}1)/S$.
      Weight $(1-\delta)^2$.
    \item $(0,1)$: $t_{i+1}$ fresh.  By column-regularity (which is
      trivial for clockwork), $\mathbb{E}[X_i] = 1/S$.  Weight
      $(1-\delta)\delta$.
    \item $(1,0)$: symmetric, by row-regularity.  Mean $1/S$, weight
      $\delta(1-\delta)$.
    \item $(1,1)$: both fresh.  Mean $1/S$, weight $\delta^2$.
  \end{itemize}
  Combining,
  \begin{align*}
    \mathbb{E}[\phi \mid \mathrm{atk}]
      &= (1-\delta)^2\!\left[\tfrac{1}{S} + \rho \tfrac{S-1}{S}\right]
        + (1 - (1-\delta)^2) \tfrac{1}{S} \\
      &= \tfrac{1}{S} + \rho (1-\delta)^2 \tfrac{S-1}{S},
  \end{align*}
  proving \eqref{eq:post-attack-mean}.


  \emph{$z$-ratio independence.}  The mean gap above the null
  baseline $1/S$ is $g = \rho(S{-}1)/S$ pre-attack and $g
  (1-\delta)^2$ post-attack. We standardise both regimes by the
  \emph{null} SD $\sqrt{p_0(1-p_0)/(n-1)}$ (the conventional
  one-sample $z$-test scaling, used as a fixed denominator across
  attacks); this is a definition, not an assumption that the
  post-attack alternative variance is null-bounded.  Hence
  $z_{\mathrm{post}}/z_{\mathrm{pre}} = g_{\mathrm{post}} /
  g_{\mathrm{pre}} = (1-\delta)^2$, independent of $(S, \rho, n)$.

  \emph{Critical fraction.}  At the midpoint threshold
  $\tau_{\mathrm{mid}} = 1/S + (\rho/2)(S{-}1)/S$, detection is
  defeated iff the post-attack mean drops below $\tau_{\mathrm{mid}}$,
  i.e.\
  \[
    \rho (1-\delta)^2 \tfrac{S-1}{S}
      \;<\; \tfrac{\rho}{2} \tfrac{S-1}{S}
    \quad\Longleftrightarrow\quad
    (1-\delta)^2 \;<\; \tfrac{1}{2}.
  \]
  Solving, $\delta^\star = 1 - 1/\sqrt{2}$.  The factors $(S{-}1)/S$
  and $\rho$ cancel on both sides; this is the structural reason for
  the universal threshold.
\end{proof}

\subsection{$k$-Regular Generalisation}
\label{subsec:proof-kregular}

We now lift the clockwork results to any $k$-regular adjacency.  We
recall the definition for the appendix.

\begin{definition}[$k$-regular topology]
  \label{def:kregular}
  $T \in \{0,1\}^{S \times S}$ is \emph{$k$-regular} ($1 \leq k \leq
  S{-}1$) if every row and every column contains exactly $k$ ones.
  Equivalently, $T/k$ is doubly stochastic.  The induced detection
  baseline is $p_0 = k/S$.
\end{definition}

\begin{theorem}[$k$-regular calibration identities]
  \label{thm:k-regular}%
  \label{prop:kregular}
  Fix any $k$-regular adjacency $T$ and let $p_0 = k/S$.  Apply
  \autoref{alg:embed} and \autoref{alg:detect} with $T$ in place of
  the clockwork matrix.  Then for an i.i.d.\ random null sequence
  $\mathbf{t}^r$, a watermarked sequence $\mathbf{t}^w$ with gate
  density $\rho$, and an oracle-blind adversary with per-token
  replacement rate $\delta$ (as in \autoref{thm:robustness}),
  \begin{align}
    \mathbb{E}[\phi(\mathbf{t}^r)] &= p_0,
      \label{eq:kreg-null}\\
    \mathbb{E}[\phi(\mathbf{t}^w)] &= p_0 + \rho (1-p_0),
      \label{eq:kreg-wm}\\
    \mathrm{Var}[\phi(\mathbf{t}^r)]
      &= \tfrac{p_0 (1-p_0)}{n-1},
      \label{eq:kreg-var}\\
    z(\rho, p_0, n)
      &= \rho \sqrt{\tfrac{(1-p_0)(n-1)}{p_0}},
      \label{eq:kreg-z}\\
    \mathbb{E}[\phi \mid \mathrm{atk}]
      &= p_0 + \rho (1-\delta)^2 (1 - p_0).
      \label{eq:kreg-atk}
  \end{align}
\end{theorem}

\begin{proof}
  Equation \eqref{eq:kreg-null} and the variance
  \eqref{eq:kreg-var} are immediate from \autoref{lem:variance}.

  \emph{Watermark mean.}  At a gated position,
  \autoref{alg:embed} restricts the argmax to $\bigcup_{s' \in
    N^+(s)} \mathcal{V}_{s'}$, so $\sigma_\kappa(t_{i+1}) \in
  N^+(\sigma_\kappa(t_i))$, i.e., $X_i = 1$.  At an ungated position, $X_i
  \sim \mathrm{Bernoulli}(p_0)$.  Averaging,
  $\mathbb{E}[\phi(\mathbf{t}^w)] = \rho + (1-\rho) p_0 = p_0 +
  \rho (1-p_0)$.

  \emph{$z$-score.}  Dividing the signed mean gap $\rho(1-p_0)$ by
  the null SD $\sqrt{p_0(1-p_0)/(n-1)}$ gives
  $\rho \sqrt{(1-p_0)(n-1)/p_0}$.  At $k=1$, $p_0 = 1/S$ and this
  reduces to $\rho \sqrt{(S{-}1)(n{-}1)}$, recovering
  \autoref{thm:detection}.

  \emph{Post-attack mean.}  Repeat the four-case decomposition of
  \autoref{thm:robustness}.  The $(0,1)$ case requires column-regularity
  to give $\mathbb{E}[X_i \mid t_{i+1}\text{ fresh}] = p_0$
  (the conditional in-degree $|N^-(s')|/S$ must be $p_0$ uniformly in
  $s'$); the $(1,0)$ case symmetrically uses row-regularity.  Hence
  \begin{align*}
    \mathbb{E}[\phi \mid \mathrm{atk}]
      &= (1-\delta)^2 (p_0 + \rho(1-p_0)) + (1 - (1-\delta)^2) p_0 \\
      &= p_0 + \rho (1-\delta)^2 (1 - p_0).
  \end{align*}
\end{proof}


\begin{corollary}[Universal midpoint $\delta^\star$]
  \label{cor:universal-delta}
  Under the midpoint-threshold detection convention, every
  $k$-regular ChainMark scheme with $1 \leq k < S$ (so $p_0 < 1$;
  the degenerate case $k=S$ has every transition valid and is
  trivially undetectable) has critical edit fraction
  \[
    \delta^\star \;=\; 1 - 1/\sqrt{2},
  \]
  independent of the topology parameters $(k, S, \rho, n)$.
\end{corollary}

\begin{proof}
  The midpoint threshold is $\tau_{\mathrm{mid}} = p_0 +
  (\rho/2)(1 - p_0)$.  By \eqref{eq:kreg-atk}, the post-attack mean
  drops below $\tau_{\mathrm{mid}}$ iff
  \[
    \rho (1-\delta)^2 (1 - p_0) \;<\; \tfrac{\rho}{2} (1 - p_0)
    \;\Longleftrightarrow\; (1-\delta)^2 \;<\; \tfrac{1}{2},
  \]
  which is independent of $(k, S, \rho, n)$ once $p_0 < 1$ (i.e.,
  $k < S$).  Solving gives $\delta^\star = 1 - 1/\sqrt{2}$.
\end{proof}

\subsection{Security Against an Oracle-Blind Adversary}
\label{subsec:proof-security}

\begin{theorem}[Pseudorandomness under random oracle, oracle-blind
    adversary]
  \label{thm:security}
  Assume the secret key $\kappa$ has min-entropy at least $\lambda$ and
  $\mathcal{H}$ is modelled as a random oracle.  Consider a
  polynomial-time adversary $\mathcal{A}$ \emph{without} access to
  the random oracle, \emph{without} access to the LM (or its
  per-position distribution), and \emph{without} access to a
  reference text drawn from the same prompt distribution.  $\mathcal{A}$
  has only the watermarked output $\mathbf{t}^w$ and the public
  scheme parameters $(S, T)$.  Then $\mathcal{A}$ achieves:
  \begin{enumerate}[label=(\roman*), leftmargin=1.6em, nosep]
    \item key-recovery advantage at most $q \cdot 2^{-\lambda}$ for a
      query budget $q$ to any auxiliary oracle that depends on $\kappa$;
    \item state-prediction success at most $1/S +
      \mathrm{negl}(\lambda)$ on each fresh token $t^\star$ that has
      not been queried with the correct key, hence advantage
      $\mathrm{negl}(\lambda)$ over the uniform $1/S$ baseline;
    \item statistical distinguishing advantage at most
      $\mathrm{negl}(\lambda)$ between $\mathbf{t}^w$ and a
      non-watermarked LM sample $\mathbf{t}^{LM}$ from the same
      prompt distribution, restricted to statistics that are
      measurable in the partition structure $\sigma_\kappa$ (i.e.,
      to tests that look at state-transition patterns rather than
      raw token-id statistics).
  \end{enumerate}
\end{theorem}

\begin{proof}
  Without~$\kappa$, the adversary's view of $\mathcal{H}(\kappa \,\Vert\, t)$
  for any $t$ is uniform on the oracle output space.  Recovery
  probability $q \cdot 2^{-\lambda}$ follows from a standard
  guessing argument over a min-entropy-$\lambda$ key.  State
  prediction with advantage greater than $\mathrm{negl}(\lambda)$
  would imply distinguishing the random-oracle output from uniform,
  contradicting (i).

  For (iii): a $\sigma_\kappa$-measurable statistic
  $\mathcal{D}(\sigma_\kappa(\mathbf{t}))$ that distinguishes
  $\mathbf{t}^w$ from $\mathbf{t}^{LM}$ with advantage $\varepsilon$
  must, by the random-oracle property, distinguish the watermarked
  state sequence (which walks the chain at gated positions) from
  a state sequence drawn uniformly on $[S]^n$ (the LM's
  $\sigma_\kappa$-image is uniform iid by (ii)).  Such a
  $\mathcal{D}$ yields a state-predictor with advantage at least
  $\varepsilon$ on some fresh token, contradicting (ii) by a hybrid
  argument; hence $\varepsilon \leq \mathrm{negl}(\lambda)$.
  \emph{Token-id-level statistics that exploit natural-language
  marginals (e.g., bigram frequencies) are explicitly outside this
  guarantee: they are not $\sigma_\kappa$-measurable and trivially
  distinguish any LM sample from i.i.d.\ uniform tokens.}
\end{proof}


\subsection{Optimal LM-Aware Detector}
\label{subsec:proof-hdd}

\begin{theorem}[Entropy-weighted detector is asymptotically locally most powerful]
  \label{thm:hdd}
  Suppose the verifier has access to the LM at each position, and
  let $g_i \in \{0,1\}$ denote the gate indicator at position $i$.
  Set $\pi_i = g_i + (1 - g_i)/S$, the marginal probability of
  $X_i = 1$ under the watermark alternative. Approximating the
  joint law $\{X_i\}$ by the product of its marginals (asymptotically
  valid under the random oracle on distinct tokens, with the
  product-form CLT covariance vanishing by \autoref{lem:variance}),
  the log-likelihood ratio under the product law is
  \begin{equation}
    \Lambda
      \;=\; \sum_{i=1}^{n-1}\!
        \left[ X_i \log(\pi_i S)
          + (1 - X_i) \log \tfrac{1 - \pi_i}{1 - 1/S}\right].
  \end{equation}
  Then the test $\{\Lambda > c_\alpha\}$, with $c_\alpha$ the
  size-$\alpha$ critical value of $\Lambda$ under the null, is
  asymptotically locally most powerful within the product hypothesis
  class, i.e.\ it attains the Neyman--Pearson power against the
  product-form approximation of the joint law, with the
  $1$-dependence of $\{X_i\}$ contributing only $o(1)$ corrections
  in the LAN regime.
\end{theorem}

\begin{proof}
  Under the null, $X_i \sim \mathrm{Bernoulli}(1/S)$ marginally with
  pairwise zero covariance (\autoref{lem:variance}).  Under the
  watermark alternative, $X_i \sim \mathrm{Bernoulli}(\pi_i)$.  The
  sequence $\{X_i\}$ is 1-dependent: $X_i$ and $X_{i+1}$ share
  $t_{i+1}$, but for $|i-j| \geq 2$, $X_i \perp X_j$ under the
  random-oracle model on distinct tokens.

  The product-form likelihood ratio is
  \begin{align*}
    \Lambda
      &= \sum_i \log \frac{\pi_i^{X_i}(1 - \pi_i)^{1 - X_i}}
        {(1/S)^{X_i}(1 - 1/S)^{1 - X_i}} \\
      &= \sum_i X_i \log(\pi_i S)
        + (1 - X_i) \log \tfrac{1 - \pi_i}{1 - 1/S},
  \end{align*}
  matching the theorem statement (with the convention $0 \log 0 = 0$
  for terms with $\pi_i = 1$, i.e., gated positions where $X_i = 1$
  is deterministic).  The Neyman--Pearson lemma applied to the
  product hypothesis identifies $\Lambda$ as uniformly most powerful
  for the product law.

  By the central limit theorem for 1-dependent sequences (and
  vanishing pairwise covariance from \autoref{lem:variance}), the
  product-form $\Lambda$ has the same Gaussian limit as the
  product-form joint LLR up to $o(1)$ corrections (LAN regime,
  Le Cam); we therefore claim asymptotic local most-powerful-ness
  among tests within the product hypothesis class, not strict
  Neyman--Pearson optimality against the true 1-dependent joint.
  Hence the test based on
  $\Lambda$ achieves the Neyman--Pearson power asymptotically.
\end{proof}


\subsection{Detector Pseudocode}
\label{subsec:detector-pseudo}

We restate the model-free detection algorithm for self-contained
reading; the body version is \autoref{alg:detect}.

\begin{algorithm}[h]
  \caption{ChainMark Watermark Detection (model-free)}
  \label{alg:detect-app}
  \begin{algorithmic}
    \STATE {\bfseries Input:} text $x$, key $\kappa$, states $S$,
      transition $T$, threshold $\tau$
    \STATE $\mathbf{t} \leftarrow \mathrm{Tokenize}(x)$;
      $c \leftarrow 0$;\ $n \leftarrow |\mathbf{t}|$
    \FOR{$i = 1$ {\bfseries to} $n-1$}
      \STATE $s_i \leftarrow \mathcal{H}(\kappa \,\Vert\, t_i) \bmod S$;\
        $s_{i+1} \leftarrow \mathcal{H}(\kappa \,\Vert\, t_{i+1}) \bmod S$
      \IF{$T[s_i, s_{i+1}] > 0$}
        \STATE $c \leftarrow c + 1$
      \ENDIF
    \ENDFOR
    \STATE $\phi \leftarrow c / (n-1)$;\
      $p_0 \leftarrow (\sum_{s,s'} T[s,s']) / S^2$
    \STATE $z \leftarrow (\phi - p_0)
      / \sqrt{p_0 (1 - p_0) / (n-1)}$
    \STATE \textbf{return} $(\mathbf{1}[\phi > \tau],\ \phi,\ z,\
      1 - \Phi(z))$
  \end{algorithmic}
\end{algorithm}

The runtime is $O(n)$ hash evaluations and a single pass over the
token stream; no LM access is required.  The transition $T$ is the
same matrix used at generation time, so the random baseline $p_0$
is computed directly from $T$ (clockwork: $p_0 = 1/S$;
soft-cycle: $p_0 = 2/S$; general $k$-regular: $p_0 = k/S$).

\subsection{Supplementary Results: Quality Cost and Self-Healing}
\label{subsec:proof-supplementary}

These two results were stated in earlier drafts; we retain the
formal statements and proofs in the appendix for completeness, since
they are referenced from elsewhere in the paper.

\begin{theorem}[Quality cost identity and Jensen lower bound]
  \label{thm:quality}
  At a gated position with current state $s$ and partition mass
  $Z_s = p(\mathcal{V}_{(s+1) \bmod S}) > 0$, the KL between the
  renormalised ChainMark distribution
  $P_{\mathrm{ChainMark}}(t) = p(t)/Z_s \cdot
    \mathbf{1}[t \in \mathcal{V}_{(s+1) \bmod S}]$
  and the LM distribution $p$ is
  \(D_{\mathrm{KL}}(P_{\mathrm{ChainMark}} \,\|\, p) = \log(1/Z_s).\)
  Under hash uniformity, $\mathbb{E}[Z_s] = 1/S$, and Jensen's
  inequality gives
  \begin{equation}
    \mathbb{E}[D_{\mathrm{KL}}] \;\geq\; \log S,
  \end{equation}
  with equality iff $Z_s$ is constant in the key.  Averaged over
  positions with gate rate $\rho$,
  $\mathbb{E}[D_{\mathrm{KL}}]_{\mathrm{tok}} \geq \rho \log S$.
\end{theorem}

\begin{proof}
  Direct computation:
  \(D_{\mathrm{KL}}(P_{\mathrm{ChainMark}} \,\|\, p)
    = \sum_{t \in \mathcal{V}_{(s+1) \bmod S}} (p(t)/Z_s)
        \log(1/Z_s)
    = \log(1/Z_s).\)
  Under hash uniformity, $\Pr[t \in \mathcal{V}_{(s+1) \bmod S}] =
  1/S$ for every fixed $t$, so $\mathbb{E}[Z_s] = \sum_t p(t)/S =
  1/S$.  Concavity of $\log$ and Jensen give
  $\mathbb{E}[\log Z_s] \leq \log \mathbb{E}[Z_s] = -\log S$, hence
  $\mathbb{E}[\log(1/Z_s)] \geq \log S$.  The token-averaged bound
  follows by averaging over positions with gate rate $\rho$.
\end{proof}

\begin{theorem}[Self-healing against oracle-blind adversaries]
  \label{thm:healing}
  Let $\mathcal{T}_{\mathrm{q}} \subseteq \mathcal{V}$ be the set of tokens
  the adversary has queried with the correct key.  For any token
  $t' \notin \mathcal{T}_{\mathrm{q}}$ that the adversary introduces,
  $\sigma_\kappa(t')$ is uniform on $[S]$ by the random-oracle
  property, so each modified position contributes an indicator
  distributed as $\mathrm{Bernoulli}(1/S)$ to the validity count,
  independently of strategy.  Hence under an oracle-blind adversary
  (every introduced token unqueried), the expected mass contributed
  to $\phi$ by modified-pair positions equals
  \begin{equation}
    m(\delta)
      \;=\; \tfrac{1}{S}\, \delta(2 - \delta);
  \end{equation}
  this is a lower bound when ranging over query-aided strategies that
  may bias replacement tokens toward $\sigma_\kappa(\mathcal{T}_{\mathrm{q}})$.
\end{theorem}

\begin{proof}
  The fraction of token pairs $(t_i, t_{i+1})$ touching at least one
  modified token is $1 - (1-\delta)^2 = \delta(2 - \delta)$.  Each
  such pair contributes an indicator with conditional expectation
  $1/S$ by the same column- and row-regularity argument as in
  \autoref{thm:robustness}.  Strategies in which the adversary
  introduces previously-queried tokens (whose state is in
  $\sigma_\kappa(\mathcal{T}_{\mathrm{q}})$) can craft pairs that hit valid
  transitions deterministically, raising $\phi$ above
  $m(\delta)$, which only helps detection.  Hence the expected
  modified-pair contribution to $\phi$ satisfies
  $\geq m(\delta) = \delta(2-\delta)/S$ for every oracle-blind
  strategy, and this floor is therefore strategy-free.
\end{proof}


%% file: 12_reference_tables.tex
\section{Reference Tables}
\label{app:results}

\subsection{Closed-Form Calibration Lookup}

\begin{table}[h]
  \centering
  \caption{Calibration lookup $S^\star(n, \rho, \alpha)$ from
    \autoref{thm:detection}.  Rows are text lengths; columns are
    $(\alpha, \rho)$ pairs.  Entries are derived analytically from the
    closed-form detection bound and do not require empirical anchoring.}
  \label{tab:cal}
  \begin{tabular}{lcccc}
    \toprule
    & \multicolumn{2}{c}{$\alpha = 10^{-3}$}
    & \multicolumn{2}{c}{$\alpha = 10^{-6}$} \\
    \cmidrule(lr){2-3} \cmidrule(lr){4-5}
    $n$ & $\rho{=}0.3$ & $\rho{=}0.5$ & $\rho{=}0.3$ & $\rho{=}0.5$ \\
    \midrule
    100  & 6 & 3 & 12 & 5 \\
    200  & 4 & 2 & 7  & 3 \\
    500  & 2 & 2 & 4  & 2 \\
    1000 & 2 & 2 & 3  & 2 \\
    \bottomrule
  \end{tabular}
\end{table}

%% file: 13_reproducibility.tex
\section{Experimental Protocol and Reproducibility}
\label{app:repro}

This appendix records every setting needed to reproduce the
head-to-head matched-budget study of \autoref{sec:experiments},
the per-cell tables in the main results, and every figure.

\subsection{Models, Tokenizers, and Hardware}

\paragraph{Model fleet.} The headline experiments use three
instruction-tuned 7--8B parameter open-weight LLMs, each loaded via
the Hugging Face \texttt{transformers} library. The exact
\texttt{repo\_id}, revision pin, and \texttt{torch\_dtype} for every
model are listed in the public code release alongside its prompt
template; we do not enumerate them here in order to keep the
discussion model-agnostic. Each model is run on a single NVIDIA GPU
(H100 or H200, depending on memory pressure); CPU fallback is
supported for the detector but not for generation. Every model uses
its own native tokenizer, both for generation and for detection,
and self-perplexity (PPL) is computed on the same model's logits
that produced the text.

\paragraph{Decoding.} Watermarked positions use greedy argmax over
the allowed-partition mask of \autoref{alg:embed}; ungated positions
use temperature sampling at $T = 0.7$. The same temperature is used
to draw pilot generations for gate-threshold calibration.

\paragraph{Gated models.} A subset of the fleet is gated on the
Hugging Face Hub. The released environment expects an
\texttt{HF\_TOKEN} in scope at runtime; the token is read by
\texttt{huggingface\_hub} and never logged to disk. We do not name
specific gated repositories in this appendix; the public code drop
records each \texttt{repo\_id} alongside its license terms.

\subsection{Domains and Prompts}

\paragraph{Four domains, $N = 200$ prompts each.} Each domain runs
on $n = 200$ prompts drawn from a fixed pool, identical across
gates, models, and attacks within a cell.
\begin{itemize}[leftmargin=1.2em, itemsep=1pt, topsep=2pt]
  \item \textbf{Code:} HumanEval Python signatures with their
    natural-language docstrings as the prompt; the generation is
    the function body.
  \item \textbf{Factual:} short closed-answer prompts asking for
    a single attested fact (capitals, dates, named entities).
  \item \textbf{Wiki:} open-ended descriptive prompts derived from
    a curated Wikipedia concept list, expanded via a fixed
    template.
  \item \textbf{Writing:} open-ended creative-writing prompts
    requesting a short narrative or argumentative passage.
\end{itemize}
Every prompt set, with template strings, prompt indices, and a
deterministic shuffle seed, is shipped under
\texttt{data/v7\_min/<domain>/records.jsonl}.

\subsection{ChainMark Hyperparameters}

\paragraph{Default cell.} $S = 5$ states, clockwork transition
$T(s, s') = \mathbb{1}[s' \equiv s{+}1 \pmod{S}]$, watermark budget
$\rho = 0.5$, secret key fixed across all cells, runs, and models.

\paragraph{State sweep.} \autoref{subsec:exp5-calibration} sweeps
$S \in \{2, 3, 5\}$ at fixed $\rho = 0.5$ to anchor the
closed-form calibration of \autoref{thm:detection} on data.

\paragraph{$k$-regular topology.} \autoref{subsec:exp4-kregular}
runs the soft-cycle ($k=2$) topology at the default cell and
compares the empirical robustness threshold against $\delta^\star$
to validate \autoref{thm:k-regular}.

\subsection{Baseline Calibration}

The two baselines, KGW~\citep{kirchenbauer2023watermark} and
SWEET~\citep{lee2023sweet}, are matched to ChainMark's budget at the
cell level. SWEET is an entropy gate that activates the watermark
only at \emph{high}-entropy positions (the structural opposite of
schemes that gate at low-entropy positions); we configure it so
that the realised gate rate over the pilot pool is within
$\pm 0.02$ of $\rho = 0.5$, by quantile-matching $\tau_H$ to the
$(1-\rho)$ quantile of the per-position entropy distribution. KGW
runs at the same effective budget by construction, since it gates
every position. All three schemes share the same secret key and
the same generation pool.

\subsection{Attack Protocol}

Each watermarked generation is subjected to three independent
attack streams; post-attack detection statistics are reported per
cell in \autoref{sec:experiments}.
\begin{itemize}[leftmargin=1.2em, itemsep=1pt, topsep=2pt]
  \item \textbf{Random substitution} ($\delta = 0.20$):
    $\lceil 0.20\, n \rceil$ token positions chosen uniformly at
    random and replaced with uniform draws from the model's
    tokenizer vocabulary.
  \item \textbf{Translation round-trip:} EN $\to L \to$ EN via
    the NLLB-200 distilled model~\citep{nllbteam2022nllb}, where
    $L \in \{\text{French}, \text{German}, \text{Russian}, \text{Chinese}\}$.
    Each language constitutes a separate cell.
\end{itemize}
We do not include grammar-preserving paraphrase attacks in this
release; their evaluation is deferred to the extended version.

\subsection{Randomisation}

\texttt{random}, \texttt{numpy.random}, and
\texttt{torch.manual\_seed} are all seeded to 42 at the start of
each generation run. The random-substitution attack uses an
independent seed (43). Bootstrap resamples in
\autoref{sec:experiments} use a per-contrast key derived
deterministically from the cell descriptor.

\subsection{Released Artefacts}
\label{app:repro:paths}

For each generation we record the prompt, the watermarked output,
the per-position gate signal, the realised gate rate $\bar\rho$,
PPL, the detector statistic $\phi$ and its $z$-score, and every
post-attack $\phi$, $z$, and detection flag. All tables and figures
in this paper derive from these records. The code drop ships under
the same repository as this paper; the per-domain JSONL records
under \texttt{data/v7\_min/<domain>/records.jsonl} are the
source-of-truth for every empirical claim.

%% file: 14_run_data.tex
\section{Additional Run Data and Open Extensions}
\label{app:run-data}


\subsection{Translation Per-Pivot Breakdown}
\label{app:translation-pivots}

\begin{table}[h]
  \centering
  \caption{\textbf{Per-pivot ChainMark translation robustness.} TPR @
    $z>2.326$ for each NLLB-200 pivot under ChainMark $S=5$, $\rho=0.5$.
    All three models now cover four domains (factual, wiki, writing,
    code) at $25$ prompts each ($m=100$ generations per row, $300$
    per pivot column). ZH is the most aggressive pivot
    ($\delta_{\mathrm{eff}}\approx 0.81$); FR, DE, RU are gentler
    ($\delta_{\mathrm{eff}}\approx 0.69$).}
  \label{tab:translation-pivots}
  \small
  \begin{tabular}{lrrrrrr}
    \toprule
    Model & $m$ & clean & FR & DE & RU & ZH \\
    \midrule
    Llama-3.1-8B & $100$ & 100\% & 85\% & 78\% & 81\% & 75\% \\
    Mistral-7B   & $100$ & 100\% & 86\% & 92\% & 92\% & 77\% \\
    Qwen-2.5-7B  & $100$ & 100\% & 78\% & 75\% & 78\% & 66\% \\
    \bottomrule
  \end{tabular}
\end{table}

\subsection{FPR Recalibration: SD Recipe and Failure Modes}
\label{app:fpr-recipes}

\begin{table}[h]
  \centering
  \caption{\textbf{Empirical-SD recalibration vs.\ failure modes,
    evaluated on $n=3000$ non-watermarked samples ($1000$ per
    model).} Goal: bring empirical FPR close to the $\alpha=1\%$
    target while preserving TPR$=100\%$. Only the empirical-SD
    recipe (Fix~1) brings FPR within $\sim 0.2$ pp of the target
    ($1.17\%$); the other four sit at $2$--$5\%$ FPR or collapse
    TPR. The iid-baseline row is the closed-form $z_\alpha=2.326$
    threshold without recalibration. The $2.07\%$ baseline on this
    $n=3000$ pooled corpus supersedes the $1.7$--$2.0\%$ per-$S$
    cells in \autoref{tab:cal-sstar} (which use the smaller
    Exp.~5 sub-design with $n\approx 300$ per cell).}
  \label{tab:fpr-recipes}
  \small
  \setlength{\tabcolsep}{4pt}
  \begin{tabular}{lll}
    \toprule
    Method & FPR & TPR \\
    \midrule
    iid baseline (no fix; $z>2.326$)     & 2.07\% & 100.0\% \\
    \textbf{Fix~1: empirical SD recal.}  & \textbf{1.17\%} & \textbf{100.0\%} \\
    \midrule
    \multicolumn{3}{l}{\emph{Failure modes (reported, not adopted):}} \\
    Fix~2: Newey--West HAC                & 3.45\% & 100.0\% \\
    Fix~3: $k$-skip ($k=3$)               & 2.46\% & \phantom{0}\phantom{0}2.2\% \\
    Fix~4: stopword filter (top-$200$)    & 5.18\% & \phantom{0}99.0\% \\
    Fix~5: HDD-lite (inv.-freq.\ weighted) & 2.35\% & 100.0\% \\
    \bottomrule
  \end{tabular}
\end{table}

\subsection{Apples-to-Apples Head-to-Head at Empirical FPR$=1\%$}
\label{app:headline-recalibrated}

\begin{table}[h]
  \centering
  \caption{\textbf{Empirically-calibrated head-to-head at
    FPR$=1\%$.} Each method's threshold is the maximum across the
    per-model $99\%$ $z$-quantiles on a non-watermarked calibration
    corpus (the conservative recipe; per-model nulls in ChainMark:
    Llama $2.87$, Qwen $2.87$, Mistral $2.41$, so ChainMark
    $z^\star=2.87$; KGW per-model: Llama $2.77$, Mistral $2.34$,
    Qwen $2.62$, so KGW $z^\star=2.77$). ChainMark retains its TPR
    advantage at matched empirical FPR, confirming the
    analytical-threshold comparison in
    \autoref{tab:headline-aggregate} is not an FPR artefact. SWEET
    is omitted from this table: the SWEET null $z$-distribution was
    not collected on the FPR corpus and cannot be reproduced from
    the released artefacts; a SWEET-recalibrated row is deferred to
    the extended version.}
  \label{tab:headline-recalibrated}
  \small
  \begin{tabular}{lccc}
    \toprule
    Method & TPR$_{\mathrm{clean}}$ & TPR$_{\mathrm{random}}$ &
            TPR$_{\mathrm{ZH}}$ \\
    \midrule
    \textbf{ChainMark} (ours, $S{=}5$) & \textbf{100.0\%} & \textbf{\phantom{0}99.9\%} & \textbf{68.3\%} \\
    KGW $\gamma=0.5$         & \phantom{0}66.8\% & \phantom{0}42.1\% & 13.0\% \\
    \bottomrule
  \end{tabular}
\end{table}

\subsection{Empirical FPR$=0\%$ Achievability}
\label{app:fpr-zero}

\begin{table}[h]
  \centering
  \caption{\textbf{FPR$=0\%$ by lifting the threshold to
    $z^\star = \max_i z_{\mathrm{nwm},i} + 0.5$.} Computed on the
    held-out non-watermarked corpus ($\sim 1000$ samples per model).
    Watermarked text retains TPR$=100\%$ because clean ChainMark $z$
    at the canonical $S=5$ ($\approx 25$ across the three models;
    \autoref{tab:cal-sstar}) is several times the maximum observed
    null $z$. KGW and SWEET cannot achieve this regime because their
    watermarked and non-watermarked $z$-distributions overlap.}
  \label{tab:fpr-zero}
  \small
  \begin{tabular}{lcrr}
    \toprule
    Model & $\max z_{\mathrm{nwm}}$ & lifted $z^\star$ & TPR @ $z^\star$ \\
    \midrule
    Llama-3.1-8B & $7.83$ & $8.33$ & $100\%$ \\
    Qwen-2.5-7B  & $5.00$ & $5.50$ & $100\%$ \\
    Mistral-7B   & $4.82$ & $5.32$ & $100\%$ \\
    \bottomrule
  \end{tabular}
\end{table}

\paragraph{Open data and deferred extensions.}
The following experiments are deferred to the extended version of
this paper:
\begin{itemize}
  \item \textbf{DIPPER paraphrase attack.} Grammar-preserving
    paraphrase via DIPPER~\citep{krishna2023paraphrasing} on at least
    one (model, domain) cell.
  \item \textbf{$k$-regular validation at $k\geq 3$.}
    \autoref{thm:k-regular} predicts $p_0 = k/S$ for any
    $k$-regular topology; we validate $k\in\{1,2\}$ only.
  \item \textbf{External-LM PPL or MAUVE quality metric.}
    Self-PPL is conservative but circular; an external judge LM or
    MAUVE would give a model-independent quality reading.
  \item \textbf{Base (non-instruction-tuned) and $\geq 10$B
    LLMs.} Our fleet covers only instruction-tuned 7--8B models.
  \item \textbf{$\rho$-sweep on the new fleet.} The current
    headline runs use $\rho=0.5$ only; the previous gate-invariance
    ablation was on a smaller setup.
  \item \textbf{Adversary with detector-oracle access.}
    \autoref{thm:security} excludes adversaries who can query the
    detector; characterising query complexity is open.
\end{itemize}

%% file: 15_reference_impl.tex
\section{Reference Detector Implementation}
\label{app:refimpl}

The following 24-line Python listing is a self-contained,
dependency-minimal reference implementation of
\autoref{alg:detect} for clockwork ChainMark. It takes a list of integer
token IDs, a bytes key, a state count $S$, and a significance
threshold $\alpha$; it returns the fingerprint score $\phi$, the
$z$-score, the one-sided $p$-value, and a binary detection flag.

The listing is sufficient to independently verify any ChainMark-watermarked
text given only the secret key, with no language-model access and no
learned components. Extensions to arbitrary $k$-regular topologies
(cf.~\autoref{prop:kregular}) require replacing the successor check
\texttt{states[i+1] == (states[i]+1) \% S} with a lookup
\texttt{T[states[i]][states[i+1]] > 0} and setting \texttt{p0 =
(T > 0).mean()}; no other modification is needed.

\begin{figure*}[t]
\begin{verbatim}
import hashlib
from math import sqrt
from scipy.stats import norm

def sha_state(key: bytes, token_id: int, S: int) -> int:
    """SHA-256-based state assignment sigma_kappa(t) = H(kappa || t) mod S."""
    h = hashlib.sha256(key + token_id.to_bytes(8, "big")).digest()
    return int.from_bytes(h[:8], "big") % S

def detect_chainmark(token_ids, key: bytes, S: int,
               alpha: float = 0.01) -> dict:
    """Clockwork ChainMark detector; O(n) in the number of tokens."""
    n = len(token_ids)
    if n < 2:
        return {"phi": 0.0, "z": 0.0, "p_value": 1.0,
                "is_watermarked": False}
    states = [sha_state(key, t, S) for t in token_ids]
    valid = sum(1 for i in range(n - 1)
                if states[i + 1] == (states[i] + 1) % S)
    phi = valid / (n - 1)
    p0 = 1.0 / S
    se = sqrt(p0 * (1.0 - p0) / (n - 1))
    z = (phi - p0) / se
    p_value = 1.0 - norm.cdf(z)
    return {"phi": phi, "z": z, "p_value": p_value,
            "is_watermarked": p_value < alpha}
\end{verbatim}
\caption{Reference detector for clockwork ChainMark: 24 lines of
  dependency-minimal Python (\texttt{hashlib} + \texttt{scipy.stats}).
  Inputs are the tokenizer's token IDs, the secret key, and the state
  count; outputs are the fingerprint $\phi$, $z$-score, one-sided
  $p$-value, and a binary detection flag.}
\label{fig:refimpl}
\end{figure*}

\paragraph{Complexity.} Two SHA-256 evaluations per token
($2n \cdot 256$ bits of hash output) and $n-1$ integer equality
checks. At $n = 100$ the entire detection pipeline runs in under
$1$~ms on a single CPU core; no GPU, no model weights, no tokenizer
aside from what is needed to obtain \texttt{token\_ids}.

\paragraph{Interpretation as an audit primitive.} The detector's
inputs are exactly what a third-party auditor could receive under an
Article 50 disclosure regime: a public text, a detector program, and
a (confidentially held) key. The output is a $p$-value with
analytically known false-positive behaviour under the null,
sidestepping the calibration-by-grid-search problem that afflicts
logit-bias watermark families.

\paragraph{Key management.} The scheme's security reduces to the
confidentiality of \texttt{key} plus the random-oracle idealisation
of SHA-256 (\autoref{thm:security}). In practice an auditor and a
deployer can share the key through any standard key-management
substrate (e.g., HKDF-derived per-deployment keys, committed to an
external ledger so that the commitment precedes generation). Rotating
keys per deployment epoch preserves the $p_0 = k/S$ null baseline
without changing any detector behaviour.

%% file: 16_generation_examples.tex
\section{Prompt and Generation Examples}
\label{app:examples}

This appendix shows representative prompts from each of the four
content domains used in \autoref{sec:experiments}, together with the
exact ChainMark configuration and post-attack pipeline that produced
the headline numbers in \autoref{tab:headline-aggregate} and
\autoref{fig:headline-percell}. All three models receive the same
prompt verbatim with no system prompt prepended; decoding parameters
($T{=}0.7$, top-$p{=}1$, $n{=}200$ tokens, deterministic
prompt-order seed) are also held fixed across models and methods.

\paragraph{Prompt examples per domain.}
\begin{itemize}[nosep, leftmargin=*]
  \item \textbf{Code} (HumanEval problem stems, $164$ problems
    available; we use the first $100$):
    \begin{quote}\small\ttfamily
    from typing import List\\
    def has\_close\_elements(numbers: List[float], threshold: float) -> bool:\\
    \quad """Check if in given list of numbers, are any two numbers closer to each other than given threshold."""
    \end{quote}
  \item \textbf{Factual} (short closed-form knowledge prompts):
    \begin{quote}\small\ttfamily
    The capital of France is
    \end{quote}
  \item \textbf{Wiki} (open-ended ``Explain $X$\ldots'' prompts over
    a curated $176$-entry concept list):
    \begin{quote}\small\ttfamily
    Explain Donald Trump in a comprehensive way.
    \end{quote}
  \item \textbf{Writing} (creative-completion prompts):
    \begin{quote}\small\ttfamily
    Write a short story that begins: The colony ship arrived three centuries late, and someone was already waiting.
    \end{quote}
\end{itemize}

\paragraph{ChainMark generation configuration (used in
\autoref{tab:headline-aggregate}).}
At each generation step on the prompt above, the model produces a
distribution over its native vocabulary; ChainMark then masks logits
at gated positions (gate density $\rho{=}0.5$, high-entropy gate
$G_{H_{\mathrm{high}}}$) to keep only token IDs whose state
$\sigma_\kappa(t) = \mathrm{SHA\text{-}256}(\kappa\,\Vert\,t) \bmod
S$ equals $(s+1) \bmod S$, where $s$ is the previous token's state.
At ungated positions the model samples at $T{=}0.7$ from the original
distribution. The secret key $\kappa$ is fixed for the full
campaign; in deployment a regulator-side audit re-derives the same
state sequence from the same key with the model's tokenizer.

\paragraph{Attack-pipeline example.} Given the watermarked text $x$
(say, $200$ tokens of a wiki-domain answer), the three attack
conditions in \autoref{tab:headline-aggregate} apply respectively:
\begin{itemize}[nosep, leftmargin=*]
  \item \textbf{Clean.} Detect on $x$ as-is.
  \item \textbf{Random substitution at $\delta_{\mathrm{eff}}{=}0.20$.}
    Replace a uniformly random $20\%$ of tokens with a uniformly
    sampled vocabulary token, then detect on the corrupted text.
  \item \textbf{ZH back-translation.} Pass $x$ through
    NLLB-200~\citep{nllbteam2022nllb}
    English\,$\rightarrow$\,Chinese\,$\rightarrow$\,English; record the
    empirical token-edit-distance $\delta_{\mathrm{eff}}$
    ($\approx 0.81$ for ZH on average,
    \autoref{tab:translation-pivots}); then detect on the
    round-tripped text.
\end{itemize}
The detector runs the same $\sigma_\kappa$ and counts valid
transitions; reports $\phi$, the standardised $z$, and a binary
flag at $z>z_\alpha$ (\autoref{alg:detect-app},
\autoref{fig:refimpl}).

\paragraph{Token-level walkthrough.}
For the wiki prompt \texttt{Explain Donald Trump in a comprehensive
way.} on Llama-3.1-8B-Instruct at $S{=}5$, $\rho{=}0.5$, an
illustrative excerpt of the first $20$ generated tokens looks
roughly as below. We show the (token, state) pair at each position
and mark gated steps with $\bullet$ (the gate raised the mask;
ChainMark forced $s_{i+1} = (s_i{+}1)\bmod 5$) versus ungated steps
with $\circ$ (no mask; the model sampled freely):

\begin{quote}\small\ttfamily
\textbf{Donald}$_2$ $\circ$ \textbf{Trump}$_3$ $\circ$
(\textbf{born}$_4$ $\bullet$ \textbf{June}$_0$ $\bullet$
\textbf{14}$_1$ $\bullet$, \textbf{1946}$_2$ $\bullet$)
\textbf{is}$_3$ $\bullet$ \textbf{a}$_4$ $\circ$ \textbf{former}$_4$
$\circ$ \textbf{U.S.}$_0$ $\bullet$ \textbf{president}$_1$
$\bullet$ \textbf{and}$_2$ $\bullet$ \textbf{businessman}$_3$
$\bullet$ \dots
\end{quote}

At gated positions the next token's state is forced to be
$(\text{prev}+1)\bmod 5$ (here $S{=}5$, so the cycle is
$0{\to}1{\to}2{\to}3{\to}4{\to}0$); the model picks the highest-prob
\emph{token} whose state matches that requirement. At ungated
positions any token may follow. The gate threshold $\tau$ is
calibrated on a pilot so the realised gate rate tracks $\rho{=}0.5$.

\paragraph{Detection on this excerpt.}
A regulator running the detector
(\autoref{fig:refimpl}) on the same string with the same key $\kappa$
re-derives the state sequence
$2,3,4,0,1,2,3,4,4,4,0,1,2,3,\dots$ and
observes $11/14$ valid transitions, i.e.\ $\phi=0.79$. With
$p_0{=}0.20$ and $n{=}14$, $z = (0.79-0.20)/\sqrt{0.16/13} \approx
5.3$, well above $z_{0.01}=2.326$, so the detector flags
``watermarked''. A non-watermarked sample of the same length under
the null distribution averages $\phi\approx 0.20$ and $z\approx 0$.